\documentclass[sigconf]{acmart}

\usepackage{caption}
\usepackage{amsmath,amsfonts}
\usepackage{amsthm}
\usepackage{algorithmic}
\usepackage{graphicx}
\usepackage{textcomp}
\usepackage{subfigure}
\usepackage{booktabs}
\usepackage{xcolor}
\usepackage{multirow}
\usepackage{array}
\usepackage{threeparttable}
\usepackage[ruled,vlined]{algorithm2e}
\usepackage{balance}
\usepackage{lipsum}
\usepackage{csquotes}
\usepackage{makecell}

\usepackage{amssymb}
\usepackage[misc]{ifsym} 

\theoremstyle{definition}
\newtheorem{definition}{Definition}

\theoremstyle{lemma}

\theoremstyle{theorem}
\newtheorem{theorem}{Theorem}

\newcommand\blfootnote[1]{%
\begingroup
\renewcommand\thefootnote{}\footnote{#1}%
\addtocounter{footnote}{-1}%
\endgroup
}

\AtBeginDocument{%
  }

\setcopyright{acmlicensed}
\copyrightyear{2018}
\acmYear{2025}
\acmDOI{XXXXXXX.XXXXXXX}

\acmConference[SIGMOD '25]{ACM Conference on Management of Data}{June 22--27,
  2025}{Berlin, Germany}

\acmISBN{978-1-4503-XXXX-X/18/06}

\begin{document}

\title{Subspace Collision: An Efficient and Accurate Framework for High-dimensional Approximate Nearest Neighbor Search}


\author{Jiuqi Wei}
\affiliation{%
  \institution{Institute of Computing Technology, Chinese Academy of Sciences}
  \institution{University of Chinese Academy of Sciences}
  \country{China}
}
\email{weijiuqi19z@ict.ac.cn}

\author{Xiaodong Lee}
\affiliation{%
  \institution{Institute of Computing Technology, Chinese Academy of Sciences}
  \institution{Fuxi Institution}
  \country{China}
}
\email{xl@ict.ac.cn}

\author{Zhenyu Liao}
\affiliation{%
  \institution{Huazhong University of Science and Technology}
  \country{China}
}
\email{zhenyu_liao@hust.edu.cn}

\author{Themis Palpanas}
\affiliation{%
  \institution{LIPADE, Universit{\'e} Paris Cit{\'e}}
  \country{France}
}
\email{themis@mi.parisdescartes.fr} 

\author{Botao Peng}
\orcid{0000-0002-1825-0097}
\affiliation{%
  \institution{Institute of Computing Technology, Chinese Academy of Sciences}
  \country{China}
}
\email{pengbotao@ict.ac.cn}


\begin{abstract}
Approximate Nearest Neighbor (ANN) search in high-dimensional Euclidean spaces is a fundamental problem with a wide range of applications.
However, there is currently no ANN method that performs well in both indexing and query answering performance, while providing rigorous theoretical guarantees for the quality of the answers.
In this paper, we first design SC-score, a metric that we show follows the \emph{Pareto principle} and can act as a proxy for the Euclidean distance between data points. 
Inspired by this, we propose a novel ANN search framework called \emph{Subspace Collision (SC)}, which can provide theoretical guarantees on the quality of its results. 
We further propose SuCo, which achieves efficient and accurate ANN search by designing a clustering-based lightweight index and query strategies for our proposed subspace collision framework.
Extensive experiments on real-world datasets demonstrate that both the indexing and query answering performance of SuCo outperform state-of-the-art ANN methods that can provide theoretical guarantees, performing 1-2 orders of magnitude faster query answering with only up to one-tenth of the index memory footprint.
Moreover, SuCo achieves top performance (best for hard datasets) even when compared to methods that do not provide theoretical guarantees.
\end{abstract}




\maketitle

\blfootnote{$^\ast$ Xiaodong Lee and Botao Peng are the corresponding authors.}

\section{Introduction}

\textbf{Background and Problem.} 
Nearest neighbor (NN) search in high-dimensional Euclidean spaces is a fundamental problem with various applications~\cite{DBLP:journals/sigmod/Palpanas15,Palpanas2019,DBLP:conf/wims/EchihabiZP20,murray2023Randomized} ranging from information retrieval~\cite{karpukhin2020dense}, data mining~\cite{tagami2017annexml}, to recommender system~\cite{schafer2007collaborative}. 
Given a dataset $\mathcal D$ of $n$ data points in $d$-dimensional space, a query $q \in \mathbb{R}^d$, and $k \in \mathbb{N}$, NN search will return the $k$ nearest neighbors of $q$ from $\mathcal D$.
However, NN search in high-dimensional datasets is challenging due to the \emph{curse of dimensionality} phenomenon~\cite{hinneburg2000nearest, weber1998quantitative, borodin1999lower,couillet2022RMT4ML,louart2018random,liao2020rff}.
To achieve a better trade-off between query efficiency and accuracy, many researchers turn to Approximate Nearest Neighbor (ANN) search, sacrificing some query accuracy to achieve huge gains in efficiency~\cite{zeyubulletin-sep23,li2019approximate,annbulletin}. 
ANN search will retrieve a set $\mathcal C$ of $k$ high-quality candidates to maximize $recall=\frac{\left| \mathcal C \cap \mathcal G \right|}{k}$, where $\mathcal G$ consists of the exact $k$ nearest neighbors of $q$ in $\mathcal D$.
Depending on where the index and dataset are stored, we can divide ANN methods into three categories: in-memory methods~\cite{malkov2018efficient,detlsh,ge2013optimized,annoy}, disk methods~\cite{liu2017pqbf,li2020efficient}, and memory-disk hybrid methods~\cite{jayaram2019diskann,chen2021spann}.
In this paper, we focus on in-memory methods due to their great success for fast high-recall search.

\textbf{Limitations and Motivation.}
Nowadays, new data is generated at an ever-increasing rate, and the size of datasets is also growing~\cite{DBLP:journals/sigmod/Palpanas15,Palpanas2019,li2024disauth,wei2023data}.
We need to manage large-scale data more efficiently to support further data analysis~\cite{hydra2,zeyubulletin-sep23,gao2023high,wang2024boldsymbol}.
Many ANN search methods have been proposed, such as Locality-Sensitive Hashing (LSH)-based~\cite{dblsh, lccslsh, pmlsh, lazylsh, andoni2015optimal,detlsh}, Vector Quantization (VQ)-based~\cite{ge2013optimized,kalantidis2014locally,norouzi2013cartesian,babenko2014inverted,xia2013joint,babenko2016efficient}, tree-based~\cite{annoy,muja2014scalable,adsconf,ulisse,peng2018paris,coconut,messi,sing,dpisax,dpisaxjournal,chatzakis2023odyssey,dumpy,seanetconf,wang2024dumpyos,leafi}, and graph-based~\cite{munoz2019hierarchical,fu2019fast,fu2021high,malkov2018efficient,lshapg,azizi2023elpis}. 
However, each method has its own advantages and disadvantages.
LSH-based methods are well-known for their robust theoretical guarantees on the quality of results, but they have to pay a high cost in terms of query answering time and have a large memory footprint~\cite{detlsh}. 
VQ-based methods use the clustering structure as the index, so the index memory footprint is small, but a lot of indexing time is required to compute fine-grained clusters~\cite{babenko2014inverted}.
Tree-based methods can partition data points into different regions by splitting nodes, but as the space dimensionality increases, the effectiveness of trees decreases~\cite{bohm2000cost, weber1998quantitative}, and query performance is thus limited.
Graph-based methods typically have better query efficiency, but take considerably more time to construct indexes and require a larger memory footprint~\cite{hydra2,li2019approximate}. 
The reason is that graph-based methods need to identify the near neighbors for each data point in the dataset (and connect to them) during the indexing phase, while in the query phase, they only need to search on a gradually converging path.
Thus, no existing ANN method performs well in both indexing and query answering performance, while providing rigorous theoretical guarantees for the quality of the answers.

\textbf{Our Method.} 
In this paper, we propose a novel ANN search framework called \emph{Subspace Collision (SC)} and design an index structure and query strategy for this framework, forming an efficient and accurate method named SuCo~\cite{sucocode}.
Compared to other ANN methods, SuCo performs well in both indexing performance (in terms of time and memory footprint) and query performance and has rigorous theoretical guarantees.
Moreover, we first design SC-score, a metric that we show follows the \emph{Pareto principle} (also known as the 80/20 rule) on many commonly used real-world datasets (Figure~\ref{scscore}), demonstrating that data points closer to the query point tend to have larger SC-scores.
As such, SC-score can act as a proxy for the Euclidean distance between data points.
Second, inspired by SC-score, we present a novel ANN search framework called \enquote{Subspace Collision (SC)} (cf. Section~\ref{scframework}) that is different from the existing LSH-based, VQ-based, tree-based, and graph-based frameworks.
We design a naive method (without index structure) called SC-Linear, which achieves extremely high recall (over 0.99) on real-world datasets (cf. Table~\ref{SC-Linear_recall}), illustrating the effectiveness of the subspace collision framework for ANN search.
Third, we propose SuCo (cf. Section~\ref{our_method}), which achieves efficient and accurate ANN search by designing an index structure and query strategies for the subspace collision framework.
SuCo constructs lightweight indexes by clustering data points in each subspace and using the inverted multi-index (IMI) to reduce the clustering complexity.
We further design a new query algorithm for IMI, called \emph{Dynamic Activation} (cf. Algorithm~\ref{dynamic_activation}), which improves the query efficiency by up to 40\% over the original query algorithm.
Fourth, we perform a rigorous mathematical analysis, showing that the proposed subspace collision framework can provide theoretical guarantees on the quality of its results (cf. Theorems~\ref{theorem1} and~\ref{theorem2}).
Fifth, we conduct extensive experiments on real-world datasets. Numerical results show that SuCo outperforms state-of-the-art ANN methods that can provide theoretical guarantees in indexing and query answering performance, performing 1-2 orders of magnitude faster query answering with only up to one-tenth of the index memory footprint.
Moreover, SuCo achieves top performance (best for hard datasets) even when compared to methods that do not provide theoretical guarantees.

Our main contributions are summarized as follows.

\begin{itemize} 
        \item We design SC-score, a metric that we show follows the \emph{Pareto principle} on commonly used real-world datasets, demonstrating that SC-score can act as a proxy for the Euclidean distance between data points. 
        \item We present \emph{Subspace Collision (SC),} a novel ANN search framework based on SC-score. Our experiments show that the subspace collision framework can achieve high recall for ANN search. We perform rigorous mathematical analysis showing that the subspace collision framework can provide theoretical guarantees on the quality of its results.
        \item We propose SuCo (code available online~\cite{sucocode}), which achieves efficient and accurate ANN search based on our subspace collision framework. 
        We design a clustering-based lightweight index 
        to ensure excellent indexing performance, and design query strategies to ensure excellent query performance.
        \item We conduct extensive experiments, demonstrating that both the indexing and query answering performance of SuCo outperforms state-of-the-art ANN methods that can provide theoretical guarantees, performing 1-2 orders of magnitude faster query answering with only up to one-tenth of the index memory footprint. Moreover, SuCo achieves top performance (best for hard datasets) even when compared to methods that do not provide theoretical guarantees.
\end{itemize}

\section{Related Work} \label{related_work}

\textbf{LSH-based methods.}
Locality-sensitive hashing (LSH)-based methods are known for their theoretical guarantees on the quality of returned query results~\cite{e2lsh, dblsh, c2lsh, qalsh, r2lsh, vhp, lccslsh, srs, pmlsh, lazylsh, eilsh, andoni2015optimal,detlsh}. 
A family of LSH functions is used to map data points from the original high-dimensional space to low-dimensional projected spaces. 
The properties of LSH can ensure that data points that are closer in the original space are also closer in the projected space~\cite{gionis1999similarity}. 
For this reason, all that is needed to obtain high-quality results is to examine the points around the query point in the projected spaces~\cite{datar2004locality}.
LSH-based methods improve query efficiency by designing the index structure and query strategy in the projected spaces.
Based on the query strategy, there are three categories of LSH-based methods: boundary constraint~\cite{e2lsh, lsbforest, sklsh, dblsh}, collision counting~\cite{c2lsh, qalsh, r2lsh, vhp, lccslsh}, and distance metric~\cite{srs, pmlsh}.
DET-LSH~\cite{detlsh} is the state-of-the-art LSH-based method, which combines the ideas of boundary constraint and distance metric.

\begin{figure*} [t!]
        \subfigcapskip=5pt
	\subfigure[Data points distribution, where the Euclidean distance from $O_1$ to $O_{10}$ to the query point $q$ gradually increases.]{
		\includegraphics[width=0.21\linewidth]{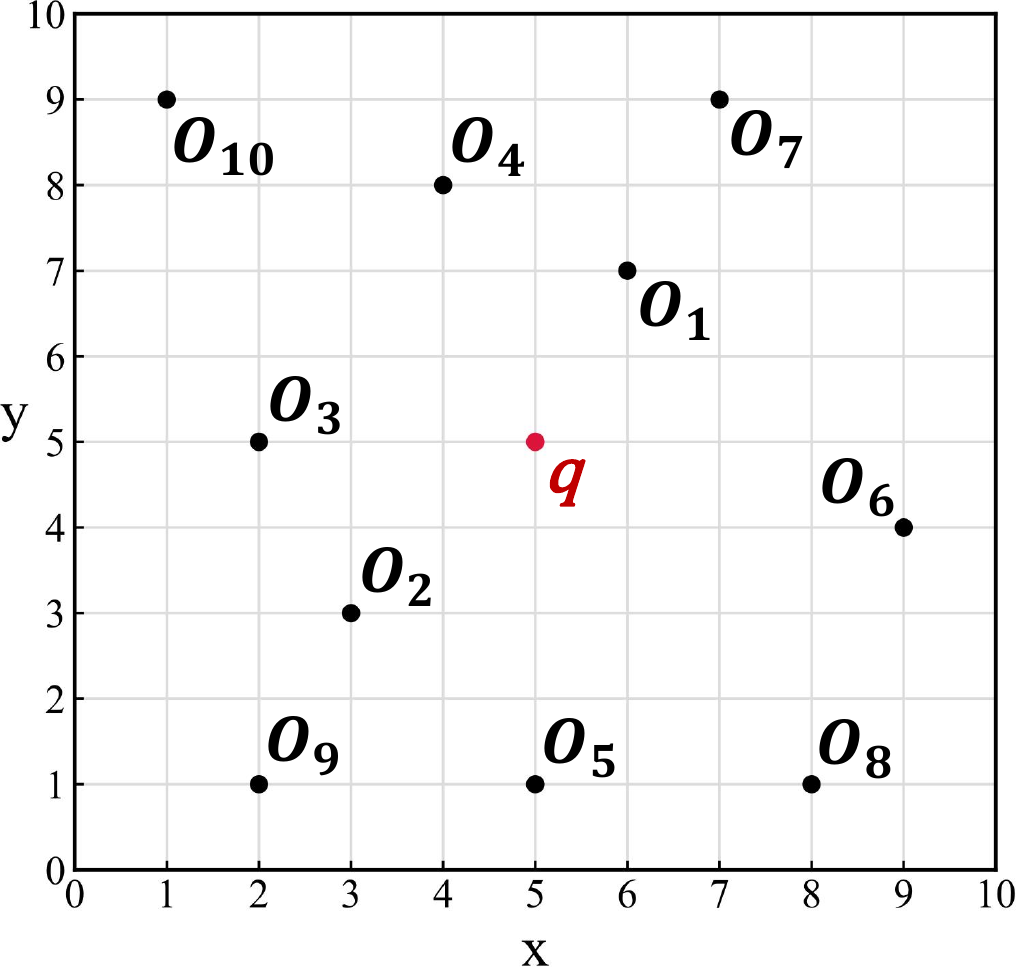}
		\label{distribution}}\hspace{7mm}
	\subfigure[Nearest neighbors when choosing $x$ or $y$ dimension as the subspace for measuring similarity.]{
		\includegraphics[width=0.31\linewidth]{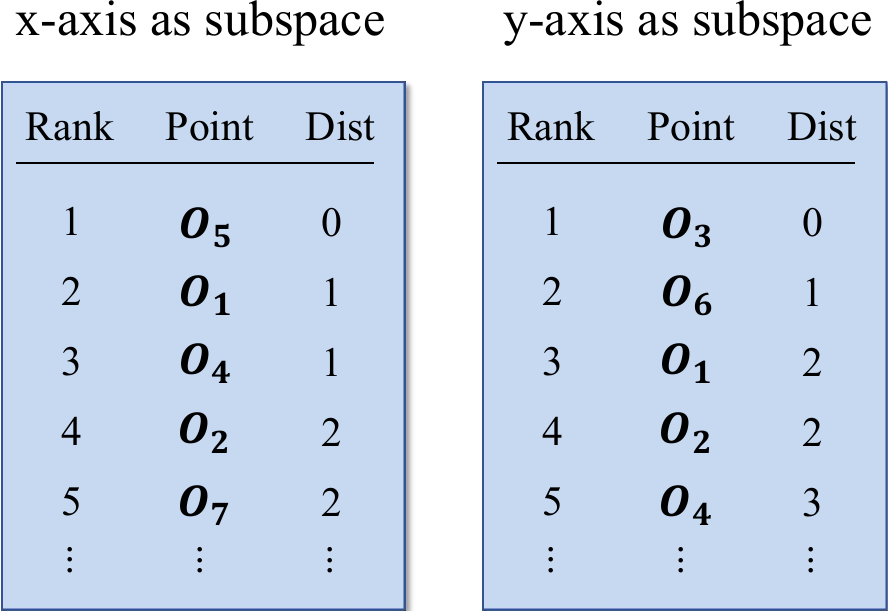}
		\label{xysubspace}}\hspace{7mm}
        \subfigure[Nearest neighbors when introducing the idea of collision based on $x$ and $y$ subspaces, where the collision ratio $\alpha=0.5$.]{
		\includegraphics[width=0.19\linewidth]{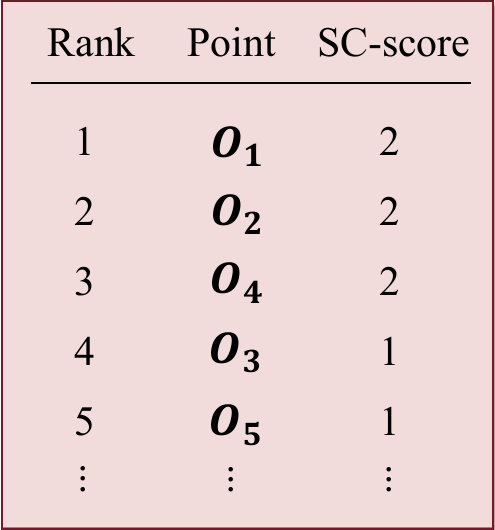}
		\label{collision}}
	\caption{Illustration of finding nearest neighbors using the idea of subspace and collision.}
	\label{subspaceandcollision}
\end{figure*}

\textbf{VQ-based methods.} 
Vector quantization is a lossy data compression technique that encodes vectors from a high-dimensional space into a finite set of values from low-dimensional discrete subspaces~\cite{gray1984vector,gray1998quantization}. 
Vector Quantization (VQ)-based methods aim to minimize the quantization distortion, which is the sum difference between each data point and its approximation~\cite{li2019approximate}.
Based on the way of combining quantizers among multiple subspaces, there are four categories of VQ-based methods: product quantization (PQ)~\cite{jegou2010product}, additive quantization (AQ)~\cite{babenko2014additive}, composite quantization (CQ)~\cite{zhang2014composite}, and tree quantization (TQ)~\cite{babenko2015tree}. 
PQ is the most popular VQ-based methods, which divide the original high-dimensional space into a Cartesian product of many low-dimensional subspaces and then quantizes subvectors in each subspace separately~\cite{matsui2018survey}.
Many extensions have been proposed to improve the performance of PQ for index construction and query answering~\cite{ge2013optimized,kalantidis2014locally,norouzi2013cartesian,babenko2014inverted,xia2013joint,babenko2016efficient,matsui2015pqtable,jegou2011searching}.
Although both PQ and subspace collision framework divide the high-dimensional space into subspaces, their design concepts are different. PQ sums the distances between quantized data points and the query in all subspaces to judge similarity, while the subspace collision framework counts the number of collisions between data points and the query in all subspaces to judge similarity.
Optimized Product Quantization (OPQ)~\cite{ge2013optimized} using the inverted multi-index~\cite{babenko2014inverted} is the state-of-the-art VQ-based method~\cite{matsui2018survey,hydra2}. 

\textbf{Tree-based methods.} 
Tree-based methods hierarchically partition the high-dimensional space and group similar data points into the same partition as leaf nodes~\cite{zeyubulletin-sep23}.
The partitioning process is usually recursive, so the root node of the tree covers the entire high-dimensional space, and its children nodes cover disjoint subspaces.
Based on the way of partitioning, there are three categories of tree-based methods: pivoting~\cite{yianilos1993data,cayton2008fast,boytsov2013learning}, hyperplane~\cite{annoy,dasgupta2008random,silpa2008optimised,muja2014scalable,iSAX2,wang2013data,ads,coconutjournal,paris+,DBLP:journals/vldb/PengFP21,ulissejournal,dpisaxjournal,sing,chatzakis2023odyssey,fatourou2023fresh,seanet,echihabi2022hercules,dumpy,leafi} and compact partitioning~\cite{fukunaga1975branch,navarro2002searching,beygelzimer2006cover}.
Hyperplane partitioning is the most popular tree-based method, which recursively partitions the space by the hyperplane with random direction or axis-aligned separating hyperplane.
Annoy~\cite{annoy} is the state-of-the-art tree-based method with balanced indexing and query answering performance~\cite{li2019approximate}.

\textbf{Graph-based methods.} 
Graph-based methods construct a proximity graph where each node represents a data point, and edges represent the neighbor relationships between data points~\cite{wang2021comprehensive,annbulletin,iliassigmod25}.
The main idea of graph-based methods is \enquote{a neighbor’s neighbor is likely also to be a neighbor}.
Compared with other types of methods, graph-based methods typically have better query efficiency but require considerably more time to construct indexes~\cite{hydra2,li2019approximate}. 
The reason is that graph-based methods need to identify the near neighbors for each data point in the dataset (and connect to them) during the indexing phase, while in the query phase, they only need to search on a gradually converging path. 
Based on the way of building the proximity graph, there are three categories of graph-based methods: cluster and merge~\cite{munoz2019hierarchical}, iterate from an initial graph~\cite{dong2011efficient,li2019approximate,fu2019fast,fu2021high,peng2023efficient,chen2018sptag}, consecutive insertion~\cite{malkov2014approximate,malkov2018efficient,lu2021hvs,lshapg,azizi2023elpis}.
Although there have been some innovative works subsequently, HNSW~\cite{malkov2018efficient} is still one of the state-of-the-art graph-based method~\cite{wang2021comprehensive}, which has an industrial-grade library and is therefore widely used.

\section{Subspace Collision Framework} \label{scframework}

\subsection{Framework Design} \label{framework_design}

Intuitively, two data points that are close in the high-dimensional original space are also more likely to be close in its random subspace. 
Given a dataset $\mathcal D$ of $n$ data points in $d$-dimensional space and a query $q \in \mathbb{R}^d$, and $o^*_1,\ldots,o^*_k$ are the $k$-NNs of $q$ in $\mathcal D$.
Suppose we randomly select $s$ dimensions from all $d$ dimensions as a subspace ($s<d$), and the points $o^*_1,\ldots,o^*_k$ and $q$ in this subspace are denoted as $o_1^{*\prime},\ldots,o_k^{*\prime}$ and $q^\prime$. 
Then, it is expected that statistically, the proximity of $o_1^{*\prime}, \ldots, o_k^{*\prime}$ and $q^\prime$ is greater than other points in $\mathcal D$.
The above intuition is based on the premise that the distances between data points are relatively evenly distributed in each dimension. 
If the distance between two points is concentrated in some dimensions, the above intuition may not hold. 
As shown in Figure~\ref{distribution} and Figure~\ref{xysubspace}, if we choose the $x$ dimension as the subspace for measuring distance, $O_5$ is the nearest neighbor because its distance from the query point $q$ is only distributed in the $y$ dimension. 
Similarly, if we choose the $y$ dimension as the subspace for measuring distance, $O_3$ and $O_6$ are the nearest neighbors since their distances to $q$ are mostly distributed in the $x$ dimension. 
Points like $O_1$ and $O_2$, which are close to $q$ but their distances evenly distributed in both $x$ and $y$ dimensions, are easily missed.

We need to design a framework to solve the above problems, which meets two requirements:
\begin{itemize} 
\item \textbf{Req1}: combine multiple subspaces to alleviate significant errors that a single subspace may cause;
\item \textbf{Req2}: reduce the impact of distance ranking within a single subspace on the final result.
\end{itemize}
To this end, we first define ``collision'' and ``subspace collision''.

\begin{definition}[Collision]\label{def_collision}
	Given a dataset $\mathcal D$ of $n$ data points in $d$-dimensional space, a query point $q \in \mathbb{R}^d$, and a collision ratio $\alpha \in \left(0, 1\right)$, if a data point $o \in \mathcal D$ satisfies: $\left\|o,q\right\|$ is one of the minimum $\alpha \cdot n$ distances between all $n$ data points and $q$, i.e., $o$ is one of the $(\alpha \cdot n)$-NNs of $q$ in $\mathcal D$, we say that $o$ collides with $q$.
\end{definition} 

\begin{definition}[Subspace Collision]\label{def_subspace_collision}
	Given a dataset $\mathcal D$ of $n$ data points in $d$-dimensional space, a query $q \in \mathbb{R}^d$, and a collision ratio $\alpha \in \left(0, 1\right)$. 
 We randomly select $s$ dimensions from all $d$ dimensions as a subspace $\mathbb{R}^s$ ($s<d$).
 The dataset $\mathcal D$, the data point $o$, and the query point $q$ in this subspace are denoted as $\mathcal D^\prime$, $o^\prime$, and $q^\prime$. If a point $o \in \mathcal D$ satisfies: $o^\prime$ is one of the $(\alpha \cdot n)$-NNs of $q^\prime$ in $\mathcal D^\prime$, we say that $o$ collides with $q$ in the subspace $\mathbb{R}^s$.
\end{definition} 

The definition of subspace collision downplays the distance ranking between data points and $q$ in the subspace, allowing \textbf{Req2} to hold.
As long as a data point is close enough to $q$ (in the $(\alpha \cdot n)$-NNs), it can be considered as colliding with $q$ in the subspace.
Subspace collision can thus alleviate the problem that the distance between $o$ and $q$ is mainly distributed in specific dimensions. 

Since using only one subspace to measure the similarity between data points and the query point may cause significant errors, we need to design a method to obtain multiple subspaces and a method that combines the collision results of these subspaces to judge the similarity.
To make full use of the information in all $d$ dimensions of the data points, we design a multi-round sampling strategy \enquote{Subspace Sampling} for the $d$ dimensions. 

\begin{definition}[Subspace Sampling]\label{def_subspace_sampling}
	Given a dataset $\mathcal D$ of $n$ data points in $d$-dimensional space, we adopt a multi-round sampling strategy to obtain $N_s$ subspaces.
 In round $i$, a number of $s=\lfloor \frac{d}{N_s} \rfloor$ dimensions are uniformly sampled without replacement to form a subspace $S_i$, $i=1,2,\ldots,N_s-1$. For the last subspace $S_{N_s}$, it simply pick up all remaining dimensions.
\end{definition}

Definition~\ref{def_subspace_sampling} makes full use of all dimensions of the data.
Even if $\frac{d}{N_s}$ is not an integer, all remaining dimensions after $N_s-1$ rounds of sampling will be picked up by the last subspace $S_{N_s}$.
For ease of reading and understanding, we assume that $\frac{d}{N_s}$ is an integer in the following sections, i.e., all subspaces ($S_1,...,S_{N_s}$) have $s=\frac{d}{N_s}$.

\begin{table}
	\centering
	\caption{Notations}
	\label{table1}
 {
	\begin{tabular}{cc}
		\toprule
		\textbf{Notation} & \textbf{Description} \\
		\midrule
		$\mathbb{R}^d$ & $d$-dimensional Euclidean space \\
		$\mathcal D$ & Dataset of points in $\mathbb{R}^d$ \\
		$n$ & Dataset cardinality  $\lvert \mathcal D \rvert$ \\
		$o, q$ & A data point in $\mathcal D$ and a query point in $\mathbb{R}^d$ \\
		$o^*_i$ & The $i$-th nearest data point to $q$ in $\mathcal D$ \\
        $o_i$ & The $i$-th data point in $\mathcal D$ \\
		$\left\|o_1, o_2\right\|$ & The Euclidean distance between $o_1$ and $o_2$ \\
        $o^{\prime},q^{\prime}$ & $o$ and $q$ in a subspace \\
        $N_s$ & Number of subspaces\\ 
        $s$ & Dimension of each subspace \\
        $S_i, \mathcal D_i$ & The $i$-th subspace and all data points in $S_i$ \\
        $o^i_j$, $q^i$ & The $j$-th data point and $q$ in the $i$-th subspace \\
        $\alpha$ & Collision ratio \\
		$\beta$ & Re-rank ratio \\
        $K, t$ & Number of K-means clusters and iterations\\
		\bottomrule
	\end{tabular}
 } 
\end{table}

\begin{algorithm}[tb]
	\caption{SC-Linear}                                                                           
	\label{SC-Linear}
	\LinesNumbered
	\KwIn{Dataset $\mathcal D$, dataset size $n$, data dimensionality $d$, query point $q$, number of results $k$, subspace number $N_s$, collision ratio $\alpha$, re-rank ratio $\beta$}
	\KwOut{$k$ nearest points to $q$ in $\mathcal D$}
        Initialize an array $SC\_scores$ of length $n$ and set all elements to 0; \\
	Divide the $d$-dimensional space into $N_s$ subspaces: $S_1, S_2,\ldots, S_{N_s}$; \\
        Divide all data points $o_1,\ldots,o_n$ and $q$ into $N_s$ subspaces; \\
	\For{$i=1$ to $N_s$}{
            \For{$j=1$ to $n$}{
                Calculate the Euclidean distance between $o^i_j$ and $q^i$; \\
            }
            Sort $o^i_1,\ldots,o^i_n$ by their distances to $q^i$ in $S_i$; \\
            \For{$z=1$ to $\alpha \cdot n$}{
                Select the $z$-th-nearest point to $q^i$ whose id in $\mathcal D$ is $t$; \\
                $SC\_scores[t]$++; \\
            }
	}
        Sort $SC\_scores$ in descending order; \\
        \For{$z=1$ to $\beta \cdot n$}{
            Select the point with the $z$-th-largest SC-score in $SC\_scores$ whose id in $\mathcal D$ is $t$; \\
            Calculate the Euclidean distance between $o_t$ and $q$; \\
        }
	\Return the \textit{top}-$k$ points closest to $q$ in the $\beta \cdot n$ candidates; \\
\end{algorithm}

Based on the obtained $N_s$ subspaces, we design a metric \enquote{Subspace Collision Score (SC-score)} to measure the similarity between a data point and the query point.

\begin{definition}[SC-score]\label{def_similarity}
	Given a dataset $\mathcal D$ of $n$ data points in $d$-dimensional space, a query $q \in \mathbb{R}^d$, $N_s$ $s$-dimensional subspaces, a collision ratio $\alpha \in \left(0, 1\right)$. Probe collisions of $q$ in $N_s$ subspaces with the collision ratio $\alpha$.
 For a data point $o \in \mathcal D$, its SC-score is the number of subspaces where it collides with $q$. Therefore, SC-score is an integer in $\left[0, N_s\right]$.
\end{definition} 

SC-score combines the collision results of $N_s$ subspaces to judge the similarity between data points and query points, allowing \textbf{Req1} to hold.
As shown in Figure~\ref{collision}, data points that are closer to the query point in the original space tend to have larger SC-scores ($\alpha=0.5$), and the subspace collision framework can return high-quality results even if the distances between data points and the query point are unevenly distributed.
For example, $O_1$ and $O_2$ are not the closest points to $q$ in both $x$ and $y$ dimensions, but they have the highest SC-score and, thus, the highest similarity.
SC-score reduces the similarity measurement error caused by the unevenly distributed distance between $o$ and $q$ in different dimensions.
The experimental results in Section~\ref{scscore_exp} (Figure~\ref{scscore}) show that SC-score is a good metric to measure the Euclidean distance between data points and query point because \enquote{\textit{data points closer to query point tend to have larger SC-scores.}}
Both Figure~\ref{subspaceandcollision} and Figure~\ref{scscore} show that the proposed subspace collision framework appropriately addresses \textbf{Req1} and \textbf{Req2} simultaneously.
Theorem~\ref{theorem1} in Section~\ref{guarantee} provides a theoretical guarantee on the effectiveness of SC-score.

\begin{figure} [tb]
        \subfigcapskip=5pt
	\subfigure[Sift10M: $d$=128, $N_s$=8, $s$=16, $\alpha$=0.1]{
		\includegraphics[width=0.46\linewidth]{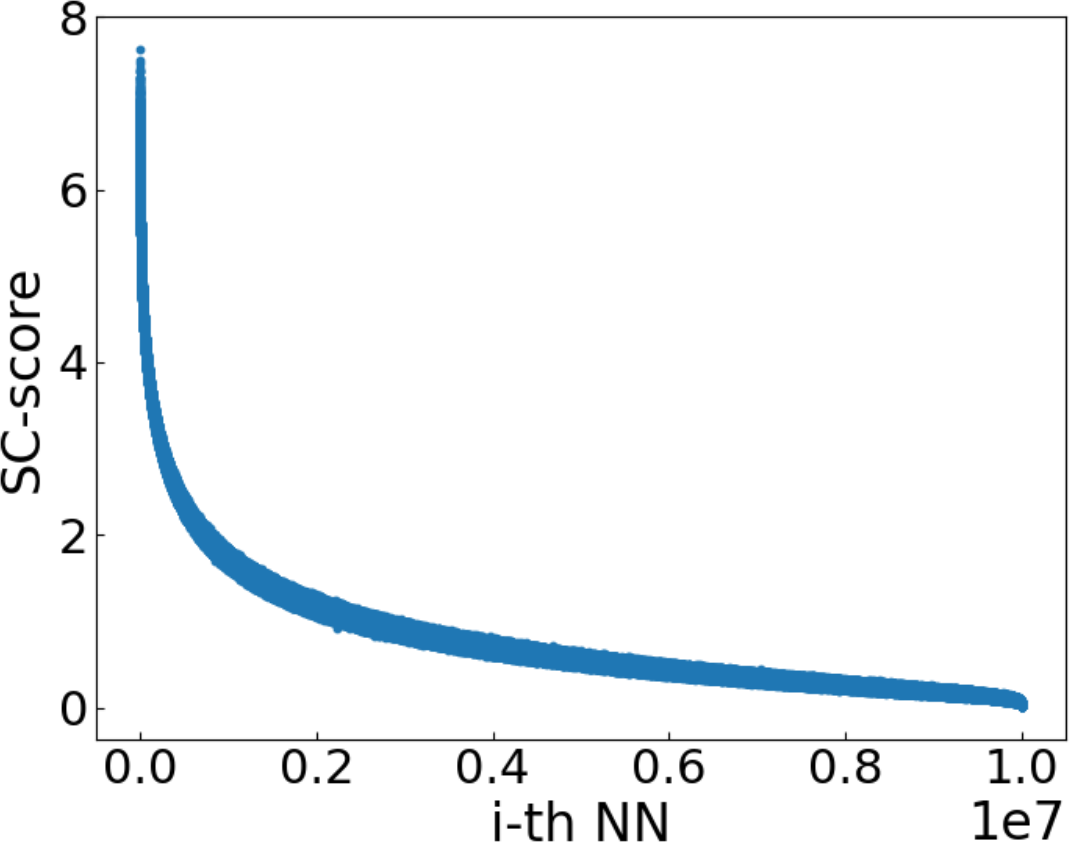}
		\label{sift10m_scscore}}\hspace{2mm}
	\subfigure[Yandex Deep10M: $d$=96, $N_s$=8, $s$=12, $\alpha$=0.1]{
		\includegraphics[width=0.46\linewidth]{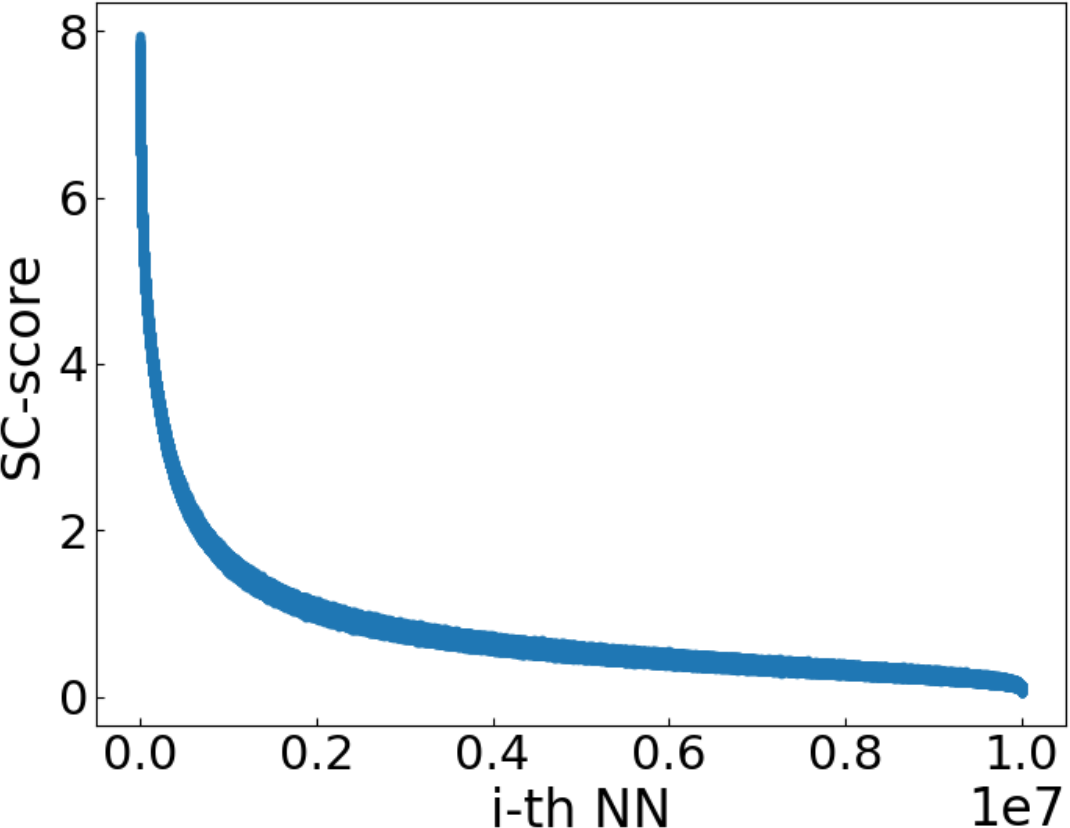}
		\label{deep10m_scscore}}\hspace{2mm}
        \subfigure[Microsoft SPACEV10M: $d$=100, $N_s$=10, $s$=10, $\alpha$=0.1]{
		\includegraphics[width=0.46\linewidth]{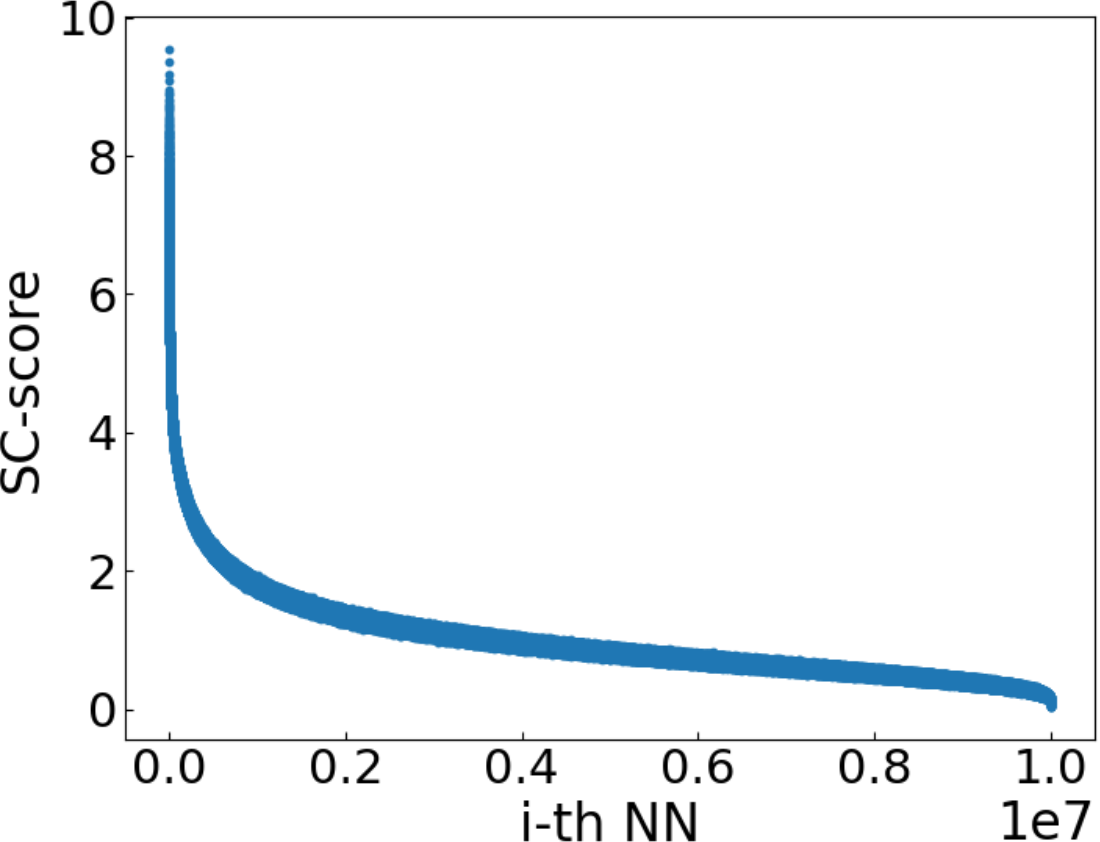}
		\label{spacev10m_scscore}}\hspace{2mm}
      \subfigure[Microsoft Turing-ANNS10M: $d$=100, $N_s$=10, $s$=10, $\alpha$=0.1]{
		\includegraphics[width=0.46\linewidth]{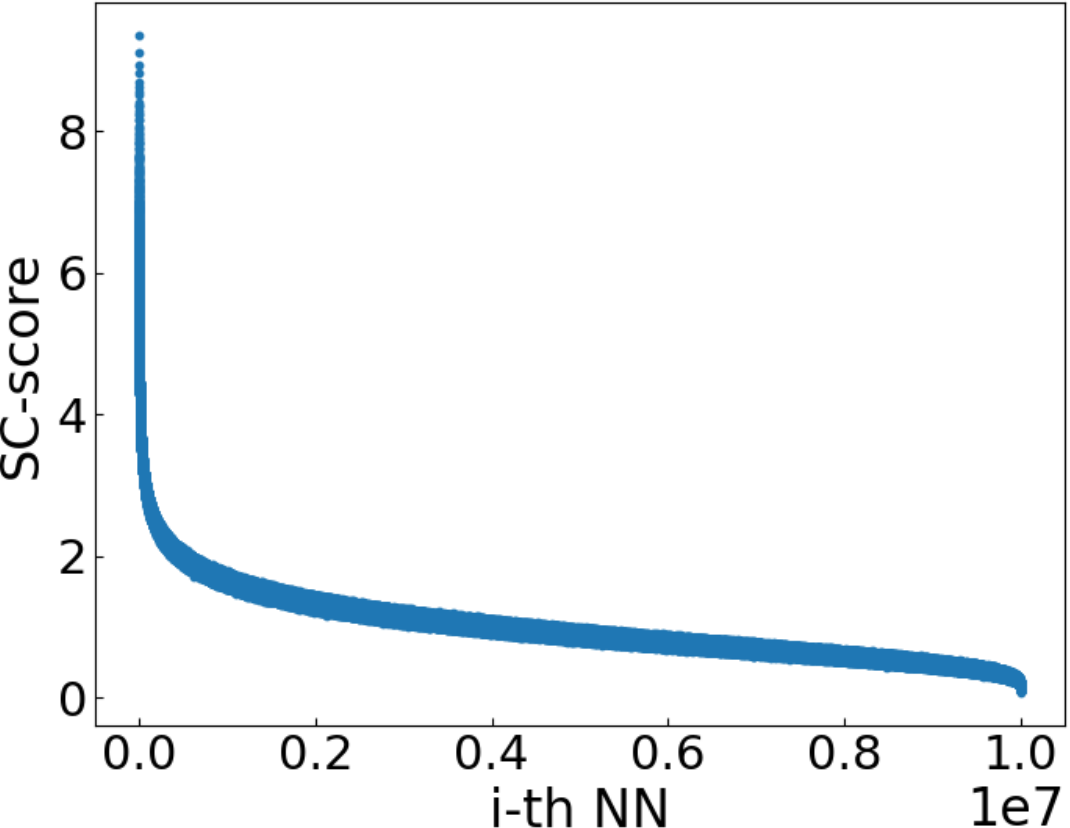}
		\label{turing10m_scscore}}
	\caption{\enquote{Pareto principle} of SC-score on four datasets. Each figure contains $n=10$M $(10^7)$ scatter points, the scatter $(i,j)$ represents the average SC-score of the $i$-th NN for 1000 queries is $j$, where $i=1,2,\ldots,10^7$, and $j \in \left[0, N_s\right]$.}
	\label{scscore}
\end{figure}

\subsection{Naive method for ANN} \label{naive_method}

Based on the subspace collision framework proposed in Section~\ref{framework_design}, we design a naive method (without index structure) called SC-Linear to support ANN search, as shown in Algorithm~\ref{SC-Linear}. 
First, an array is initialized to record the SC-score of each data point (line 1). 
Then, use Definition~\ref{def_subspace_sampling} to sample and obtain $N_s$ subspaces, and divide all data points and the query point into these subspaces, where $o^i_j$ and $q^i$ represent the $j$-th data point and $q$ in the $i$-th subspace (lines 2-3).
In practice, for convenience, we can divide the $d$-dimensional space into $N_s$ subspaces of the same length $s=\frac{d}{N_s}$, where the $i$-th subspace is allocated from the ($s \cdot (i-1)+1$)-th to ($s \cdot i$)-th dimensions in the original space, $i=1,\ldots,N_s$.
This division method can be regarded as a special case of Definition~\ref{def_subspace_sampling}.
In order to count the SC-score of each data point, it is necessary to calculate the Euclidean distance between each data point and $q$ in each subspace (lines 4-6). 
For each subspace $S_i$, sort the distances and select $\alpha \cdot n$ points closest to $q^i$ as collisions, and increase the SC-score of these points by one, where $\alpha$ is the collision ratio (lines 7-10).
To further improve query accuracy, we design a re-ranking mechanism (lines 11-15).
Specifically, sort the SC-scores of all points and select $\beta \cdot n$ points with the largest SC-scores as candidates, where $\beta$ is the re-rank ratio (lines 11-13). Calculate the Euclidean distance between all candidates and $q$ in the original space, and return the \textit{top}-$k$ points closest to $q$ (lines 14-15).
To facilitate understanding of SC-Linear, readers can refer to Figure~\ref{overview} (the workflow of SuCo method). The difference between SC-Linear and SuCo is that SC-Linear does not construct the inverted multi-indexes to speed up query answering.

The experimental results in Section~\ref{SC-Linear_performance} show that SC-Linear has high query accuracy, which demonstrates the effectiveness of our designed subspace collision framework. 
Theorem~\ref{theorem2} in Section~\ref{guarantee} provides a theoretical guarantee on the quality guarantee of results returned by Algorithm~\ref{SC-Linear}.
However, SC-Linear has no index structure to achieve fast query answering. It relies on calculating the Euclidean distance between all data points and the query point in all subspaces, which has the same time complexity as linear scanning in the original space (that is why we named Algorithm~\ref{SC-Linear} SC-Linear). 
Therefore, we need to further design an index structure and query strategies for the subspace collision framework to improve query efficiency, which will be introduced in Section~\ref{our_method}.

\begin{table}[tb]
\centering
	\caption{Recall of SC-Linear when returning 50 NNs under different datasets and parameters.}
	\label{SC-Linear_recall}
\begin{tabular}{c m{1.2cm}<{\centering}m{1.2cm}<{\centering}m{1.2cm}<{\centering}m{1.2cm}<{\centering}}
\hline
                & \makecell[c]{$\alpha$=0.05 \\ $\beta$=0.001} & \makecell[c]{$\alpha$=0.05 \\ $\beta$=0.005} & \makecell[c]{$\alpha$=0.05 \\ $\beta$=0.01} & \makecell[c]{$\alpha$=0.05 \\ $\beta$=0.05} \\ \hline
\textit{Sift10M} & 0.9536  & 0.9916         & 0.9968                & 1                   \\ 
\textit{\makecell[c]{Deep10M} }  & 0.9418  & 0.9848         & 0.9968              & 1                   \\ 
\textit{\makecell[c]{SPACEV10M} }  & 0.9742  & 0.994         & 0.9962              & 0.9992                   \\ 
\textit{\makecell[c]{Turing10M} }   & 0.9638  & 0.9876              & 0.9944               & 0.999                   \\ \hline
\end{tabular}
\end{table}

\subsection{Preliminary Experiments} \label{preliminary_experiments}

In this section, we conduct preliminary experiments to determine whether the subspace collision framework is effective. 
On the one hand, we explore whether SC-score is a good measure of the distance between all data points and the query point. 
On the other hand, we measure the performance of the SC-Linear algorithm to see whether it can return high-quality query results. 
More detailed experiments and parameter learning will be given in Section~\ref{selfevaluation}.

\subsubsection{Effectiveness of SC-score} \label{scscore_exp}

We select 1000 queries and count the SC-score of all data points in the dataset under each query. 
We perform this statistical analysis on multiple commonly used ANN datasets. 
Due to space limitations, we only show the results on four datasets in Figure~\ref{scscore} (other results are similar).

We can see from Figure~\ref{scscore} that SC-score follows the \enquote{Pareto principle} and is \textit{insensitive} to the data distribution.
Data points that are close to the query point have a high SC-score. 
Then, as the distance increases, the SC-score decreases rapidly until a \enquote{turning point} is reached. 
After the \enquote{turning point}, the SC-score decreases slowly with the distance. 
The sharp change in trend makes each figure look like an \enquote{L} shape.

The \enquote{Pareto principle} of SC-score makes it very suitable for ANN search. 
In all four datasets, the x-axis of the \enquote{turning point} is about $0.2n$, which means that the 20\% of data points closest to the query have a large and distinguishable SC-score, while the remaining 80\% of the data points have a small and indistinguishable SC-score. 
Therefore, we only need to select data points with a large enough SC-score before the \enquote{turning point} to obtain high-quality ANN search results.
SC-Linear is designed based on this idea. 
Specifically, the re-ranking mechanism (lines 11-15 in Algorithm~\ref{SC-Linear}) selects $\beta \cdot n$ data points with the largest SC-scores as candidates and then returns the \textit{top}-$k$ data points closest to the query among these candidates. 
Experimental results in Section~\ref{SC-Linear_performance} demonstrate that SC-Linear can return high-quality ANN results.

\subsubsection{Performance of SC-Linear} \label{SC-Linear_performance}

To measure whether SC-Linear can return high-quality results, we perform ANN search when $k$=50 on four datasets ($N_s$ and $s$ set for each dataset are the same as that in Figure~\ref{scscore}). 
Table~\ref{SC-Linear_recall} shows the average recall of 100 queries. We can see that SC-Linear can support accurate ANN search with extremely high recall, demonstrating that the subspace collision framework we designed is very suitable for ANN search. 
Furthermore, when the parameters $N_s$ and $\alpha$ are set reasonably, SC-Linear can still have a high recall even if a very small re-rank ratio $\beta$ is used. 
This also confirms the rule that \enquote{\textit{data points closer to query point tend to have larger SC-scores}}.
More detailed parameter learning on $\alpha$, $\beta$, and $N_s$ will be given in Section~\ref{selfevaluation}.


\subsection{Theoretical Guarantee} \label{guarantee}

In this section, through rigorous mathematical analysis, we provide theoretical guarantees on the effectiveness of SC-score and the quality guarantee of our proposed subspace collision framework for ANN search. 

\begin{theorem}[Effectiveness of SC-score]\label{theorem1}
Given a query point $q \in \mathbb{R}^d$, two independent random data points $o_1, o_2 \in$ dataset $\mathcal D$, and the SC-score of $o_1$ is greater than that of $o_2$, then $\left\|o_1, q\right\| < \left\|o_2, q\right\|$ holds with probability at least $1/2 - 1/e^2$ for appropriate choices of the number of subspaces $N_s$ and collision ratio $\alpha \in (0,1)$ that depend on the data statistics and the dimension $d$.
\end{theorem}

\begin{proof}[Proof of Theorem~\ref{theorem1}]
In the following, we provide a proof of Theorem~\ref{theorem1} by making explicit the impact of the number of subspaces $N_s$ and the collision ratio $\alpha \in (0,1)$ as a function of the data· statistics.
Assume without loss of generality that $\frac{d}{N_s}$ is an integer.
Let $o_i^j \in \mathbb{R}^{\frac{d}{N_s}}$ and $q^j$ denote the subvector of $o_i$ and $q$ in the $j$-th subspace, with $i \in \{1, 2\}$, $j \in \{1,\ldots,N_s\}$, and $N_s$ is the number of subspaces.
Let $z_i= |o_i-q| \in \mathbb{R}^d$ denote the vector of absolute difference between $o_i$ and $q$, and $z_i^j \in \mathbb R^{\frac{d}{N_s}}$ its subvectors. 
Assume that the squared Euclidean norms of subvector $Z_i^j \equiv \left\|z_i^j\right\|^2$ are independent random variables with mean $m > 0$ and variance $\sigma^2$,
we compare the following differences
\begin{equation}
\textstyle \left\|z_1\right\|^2=\sum_{j=1}^{N_s}Z_1^j, \quad \left\|z_2\right\|^2=\sum_{j=1}^{N_s}Z_2^j,
\end{equation} 
by evaluating the number of subspace collisions (i.e., SC-score).
Note that the subvectors $o_1^j$ versus $o_2^j$ with same index $j$ \emph{must} belong to one the following three scenarios: 
\begin{itemize}
    \item[(i)] ``subspace collision'' happens for both subvectors, this happens $C \in \{0,\ldots,N_s-1\}$ times; or
    \item[(ii)] ``subspace collision'' happens for $o_1^j$ but not $o_2^j$, this happens $\Delta \in \{1,\ldots,N_s - C\}$ times; or
    \item[(iii)] ``subspace collision'' happens for neither subvector, this happens $N_s - C - \Delta$ times.
\end{itemize}
For scenario~(i), it follows from the \textit{Paley–Zygmund anti-concentration inequality} that,
\begin{equation}
    \textstyle \Pr \left( Z_i^j \leq \sqrt{ (1 - \alpha) (\sigma^2 + m^2) } \right) \leq \alpha,
\end{equation}
with a collision ratio $\alpha \in \left(0, 1\right)$.
We thus have, on the index set $\mathcal{S}_C$ (of cardinality $C$) on which collision occurs, that
\begin{equation}
\textstyle 0 \le \sum_{j \in \mathcal{S}_C} Z_i^j \le C \sqrt{ (1 - \alpha) (\sigma^2 + m^2) }.
\end{equation} 
For scenario~(iii), on non-collision subspaces (that corresponds to the index set $\mathcal{S}_{\neg C}$ having cardinality $N_s-C-\Delta$), note that the random variable $Z_1^j-Z_2^j$ is of zero mean and variance $2 \sigma^2$, it then follows from \textit{Chebyshev's inequality} that, for any $t>0$,
\begin{equation}
\Pr{(|Z_1^j-Z_2^j| \ge t)} \le 2\sigma^2/t^2,
\end{equation}
and thus
\begin{equation}
\textstyle 0 \le \sum_{j \in \mathcal{S}_{\neg C}} |Z_1^j-Z_2^j| \le (N_s-C-\Delta)t.
\end{equation}
For scenario~(ii), we have, for some $t' >0$ and a given $j$,
\begin{equation*}
\textstyle 0 \le Z_1^j \le \sqrt{ (1 - \alpha) (\sigma^2 + m^2) }, \quad \Pr{(|Z_2^j-m| \ge t')} \le (\sigma/t')^2.
\end{equation*}
Finally, adding things up leads to
\begin{equation}
\textstyle  \sum_{j \in \mathcal{S}_C \cup \mathcal{S}_{\neg C}}  |Z_2^j-Z_1^j| \le C \sqrt{ (1 - \alpha) (\sigma^2 + m^2) } + (N_s-C-\Delta)t,
\end{equation}
holds with probability at least $1 - 2(N_s - C - \Delta) \sigma^2/t^2$, for any $t > 0$.
Note that the \textit{worst case} is the case of $\Delta = 1$, taking then $t = c_1 (m - \sqrt{ (1 - \alpha) (\sigma^2 + m^2) } )$ and $t' = c_2 m+ (1 - c_2) \sqrt{ (1 - \alpha) (\sigma^2 + m^2) })$ with some $c_1, c_2 >0$ such that
\begin{equation}
\begin{aligned}
    (c_2 - c_1 & (N_s - C - 1)) m > \\
    & \left( C + c_2 - c_1 (N_s - C - 1) \right) \sqrt{ (1 - \alpha) (\sigma^2 + m^2) },
\end{aligned}
\end{equation}
then we have $\left\|o_1, q\right\| < \left\|o_2, q\right\|$.
It can be checked that this happens with probability at least $ 1 - \frac{2 (N_s - 1)}{c_1^2} \left(m/\sigma - \sqrt{ (1 - \alpha) (1 + m^2/\sigma^2) } \right)^{ - 2} - \left( c_2 \cdot m/\sigma + \sqrt{ (1 - \alpha) (1 + m^2/\sigma^2) } (1-c_2) \right)^{-2}$.
For given $m, \sigma^2$, we take $c_1 = \sqrt{ 8 (N_s -1) }/ (m/\sigma - \sqrt{ (1 - \alpha) (1 + m^2/\sigma^2) } )$, $c_2 = (e  - \sqrt{  (1 - \alpha) (1 + m^2/\sigma^2) })/( m/\sigma - \sqrt{  (1 - \alpha) (1 + m^2/\sigma^2) } ) $ and any $\alpha > \max (1/(1 + m^2/\sigma^2), 1 - e^2/(1 + m^2/\sigma^2))$, we have $c_1, c_2 > 0$ so that the success probability is at least $1/2 - 1/e^2$.
This concludes the proof of Theorem~\ref{theorem1}.
\end{proof}


\begin{theorem}[Quality Guarantee of ANN Search]\label{theorem2}
Given a query point $q \in \mathbb{R}^d$ and a dataset $\mathcal D$ consist of $n$ independent random vectors, then,
Algorithm~\ref{SC-Linear} can answer a $k$-ANN query for $q$ with probability at least $1/2$, with appropriate choices of the number of subspaces $N_s$, the collision ratio $\alpha \in (0,1)$, and the re-rank ratio $\beta \in (0,1)$, as a function of the data statistics, the dataset size $n$, and the dimension $d$.
\end{theorem}

\begin{proof}[Proof of Theorem~\ref{theorem2}]
Without loss of generality, we assume $d/N_s$ is an integer.
Following the notations in the proof of Theorem~\ref{theorem1}, we denote $o_i^j \in \mathbb{R}^{\frac{d}{N_s}}$ and $q^j$ the subvector of $o_i$ and $q$ in the $j$-th subspace, respectively, with $i \in \{1, \ldots, n \}$ and $j \in \{1,\ldots,N_s\}$.
Let $z_i= |o_i-q| \in \mathbb{R}^d$ denote the vector of absolute difference between $o_i$ and $q$, and $z_i^j \in \mathbb R^{\frac{d}{N_s}}$ its subvectors. 
Assume that $Z_i^j \equiv \left\|z_i^j\right\|^2$ are independent random variables across index $i$ and $j$, with mean $m > 0$ and variance $\sigma^2$, we obtain by independence that
\begin{equation}
    \textstyle \mathbb{E}[\| z_i \|^2] = N_s m, \quad {\rm Var}[\| z_i \|^2] =  N_s \sigma^2,
\end{equation}
for which each index $j$ \textit{must} belong to one the following two cases:
\begin{itemize}
    \item[(i)] ``subspace collision'' happens for subspace $o_i^j$, this happens for $C \in \{0,\ldots,N_s\}$ times for the data vector $o_i$; or
    \item[(ii)] ``subspace collision'' does not happen for $o_i^j$, this happens for  $N_s - C $ times for the data vector $o_i$.
\end{itemize}
Similar to the proof of Theorem~\ref{theorem1}, in each case, we can bound the value of $Z_i^j$ as follows.
\begin{itemize}
    \item[(i)] By \textit{Paley–Zygmund inequality} that 
    \begin{equation}
        \textstyle \Pr \left( Z_i^j \leq \sqrt{ (1 - \alpha) (\sigma^2 + m^2) } \right) \leq \alpha,
    \end{equation}
    with a collision ratio $\alpha \in \left(0, 1\right)$.
    \item[(ii)] By \textit{Chebyshev's inequality} that, for any $t>0$,
    \begin{equation}
    \Pr{(|Z_i^j-m| \ge t)} \le \sigma^2/t^2.
    \end{equation}
\end{itemize}
We thus have that
\begin{equation*}
    (N_s - C) (m -t) \leq \| z_i \|^2 \leq C \sqrt{ (1 - \alpha) (\sigma^2 + m^2) } + (N_s - C) (m+t),
\end{equation*}
holds with probability at least $1 - (N_s - C)\sigma^2/t^2$.
In the following, without loss of generality, we only discuss here the case where the re-rank ratio $\beta$ is chosen so that $C = N_s$, for which we should have $k$ and $\beta$ both small, and $\| z_i \|^2 \leq N_s \sqrt{ (1 - \alpha) (\sigma^2 + m^2) }$.
Other scenarios with $C < N_s$ can be similarly studied by increasing $\beta$ accordingly.
In this setting, we show in the following, that the obtained range of the squared Euclidean distance $\| z_i \|^2 \in (0, N_s \sqrt{ (1 - \alpha) (\sigma^2 + m^2) }]$ is indeed among the top-$k$ smallest in a set of $n$ samples, with a controlled probability.
The moments and joint distribution of the top-$k$ ordered (e.g., smallest) values in a set of $n$ i.i.d.\@ samples has been extensively studied in the literature of \emph{order statistics}~\cite{2003Bounds,2003Expected}.

For the simplicity of exposition, we discuss here the case where the squared distances are normally distributed $\| z_i \|^2 \sim \mathcal{N}(N_s m,N_s \sigma^2)$ (which, as we shall see below, admits closed-form approximation for moments and leads to an explicit control of the success probability as a function of all tuning parameters).
To treat general (e.g., non-Gaussian) distributions, it suffices to replace the (approximations of) first and second moments in \eqref{eq:approx_expectation_gaussian}~and~\eqref{eq:approx_variance_gaussian} using PDF and CDF of the specific distribution.
In the normal case, the expectation of the $k$-th order statistic (i.e., the expected value of the $k$ smallest value) in a set of $n$ sample with $n$ large is approximately given by
\begin{equation}\label{eq:approx_expectation_gaussian}
    E_{k,n} = N_s m + \sqrt{ N_s \sigma^2 } \cdot \Phi^{-1} \left( \frac{k - \gamma}{ n - 2 \gamma + 1 } \right), \quad \gamma = 0.375,
\end{equation}
and variance approximately given by 
\begin{equation}\label{eq:approx_variance_gaussian}
    V_{k,n} = N_s \sigma^2 \cdot \frac{k (n-k+1)}{(n+1)^2 (n+2)} \left( \phi \left( \Phi^{-1} \left( \frac{k}{n+1} \right) \right) \right)^{-2},
\end{equation}
with $\phi(\cdot)$ and $\Phi(\cdot)$ the PDF and CDF of standard Gaussian distribution.
These approximations are known to be rather accurate as long as $n \geq 9$, see~\cite{blom1958Statistical,baglivo2005Mathematica}.
Thus, it follows again from \textit{Chebyshev's inequality} that the probability of correct answering is at least $1 - V_{k,n}/t^2$ for $t$ such that $t > N_s m \sqrt{ (1-\alpha) (1+ \sigma^2/m^2) } - E_{k,n}$.
Taking $N_s$ and $\alpha \in (0,1)$ and $t = \sqrt{2 N_s} \frac{\sigma}{n} \frac{k(n - k +1)}{n} \left( \phi \left( \Phi^{-1} \left( \frac{k}{n+1} \right) \right) \right)^{-1} $, we conclude that the probability of correct answering is at least $1/2$.
This concludes the proof of Theorem~\ref{theorem2}.
\end{proof}

The proof of Theorem~\ref{theorem1} relies on the idea that in cases where ``subspace collision'' happens for both or neither of the two subvectors, their Euclidean distance should be approximately the same; while in cases where ``subspace collision'' happens for one but not the other, their distance should differ by a significant amount, making the overall distance different.
A similar argument holds for a single data in the proof of Theorem~\ref{theorem2}.
Intuitively, larger values of collision ratio $\alpha$ and re-rank ratio $\beta$ lead to higher computational complexity, while too small $\alpha$ and $\beta$ reduce the accuracy of the proposed approach.
As a consequence, the parameters $\alpha$ and $\beta$ should be chosen (within a range) to achieve an optimal ``computation-accuracy trade-off.''
This is consistent with the choice of $\alpha,\beta$ in the proofs of Theorems~\ref{theorem1}~and~\ref{theorem2} above, and with our experimental conclusions in Section~\ref{sec:parameter_study}.

\section{The SuCo Method} \label{our_method}

As analyzed in Section~\ref{naive_method} and Section~\ref{preliminary_experiments}, SC-Linear's query efficiency is limited because when counting collisions in each subspace, it is necessary to calculate the Euclidean distance between the query and all data points (lines 4-10 in Algorithm~\ref{SC-Linear}). 
Therefore, the problem that needs to be solved is transformed into \enquote{\textit{how to find the $\alpha \cdot n$ data points closest to the query in each subspace as quickly and accurately as possible.}}

In this section, we propose SuCo to accelerate collision counting in each subspace by building lightweight indexes and designing query strategies. 
Figure~\ref{overview} provides a high-level overview of the SuCo workflow, including index construction and query answering. 

\begin{figure} [tb]
	\subfigcapskip=5pt
	\subfigure[Index construction.]{
		\includegraphics[width=0.95\linewidth]{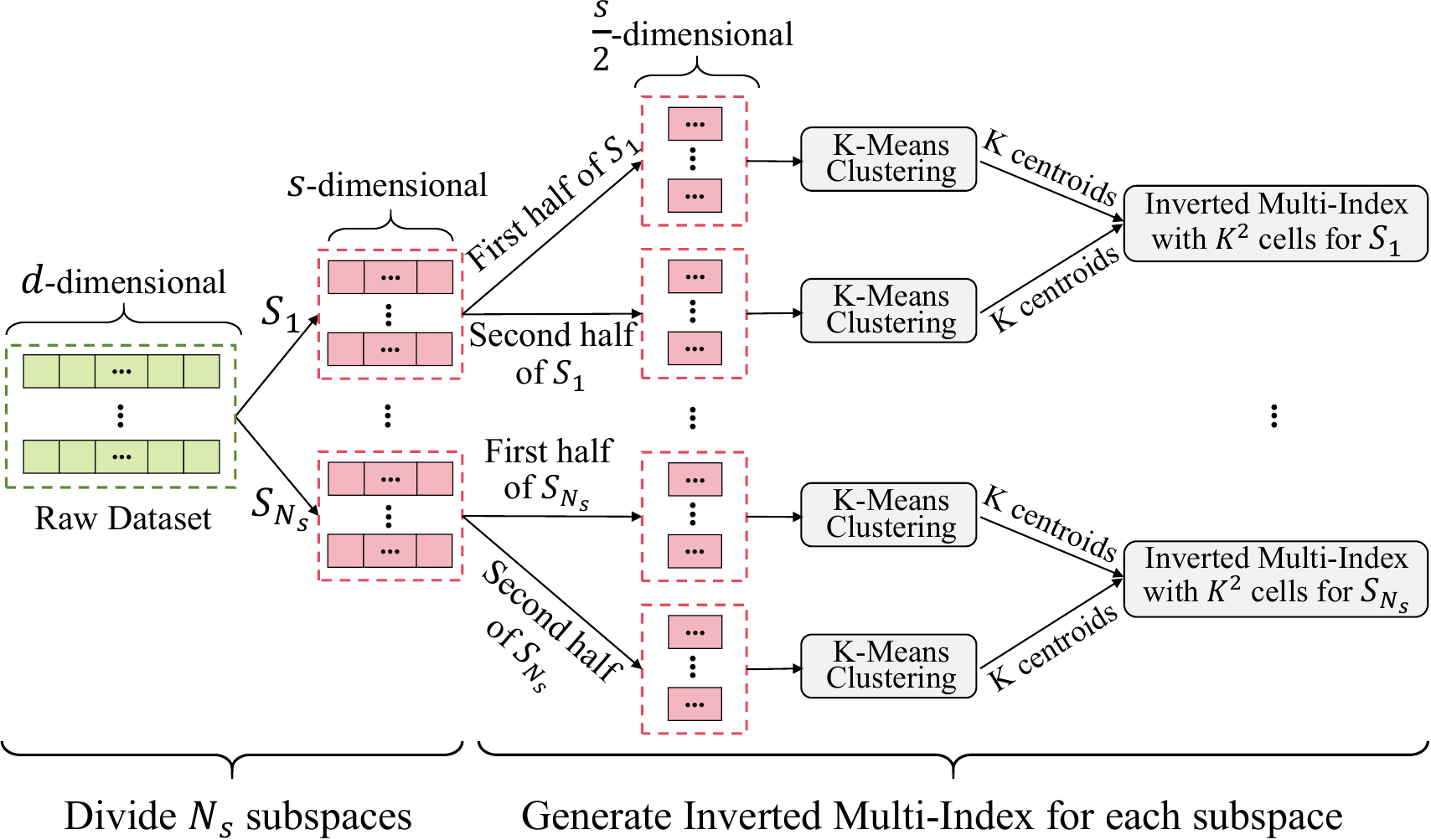}
		\label{overview_index}}
	\subfigure[Query answering.]{
		\includegraphics[width=0.95\linewidth]{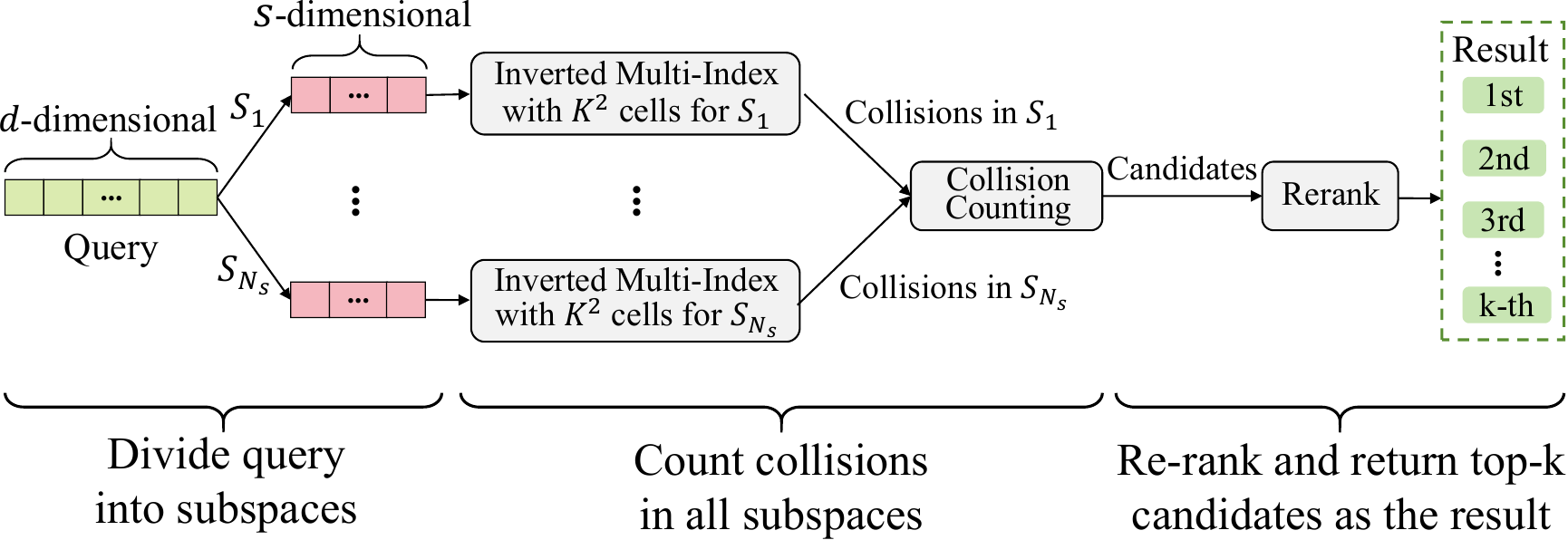}
		\label{overview_query}}
	\caption{Overview of the SuCo workflow.}
	\label{overview}
\end{figure}

\begin{figure*} [tb]
	\subfigcapskip=5pt
	\subfigure[K-means clustering with inverted index. K-means clustering is performed in the 2D space. The inverted index holds $K$ 2D centroids (red stars in the figure).]{
		\includegraphics[width=0.38\linewidth]{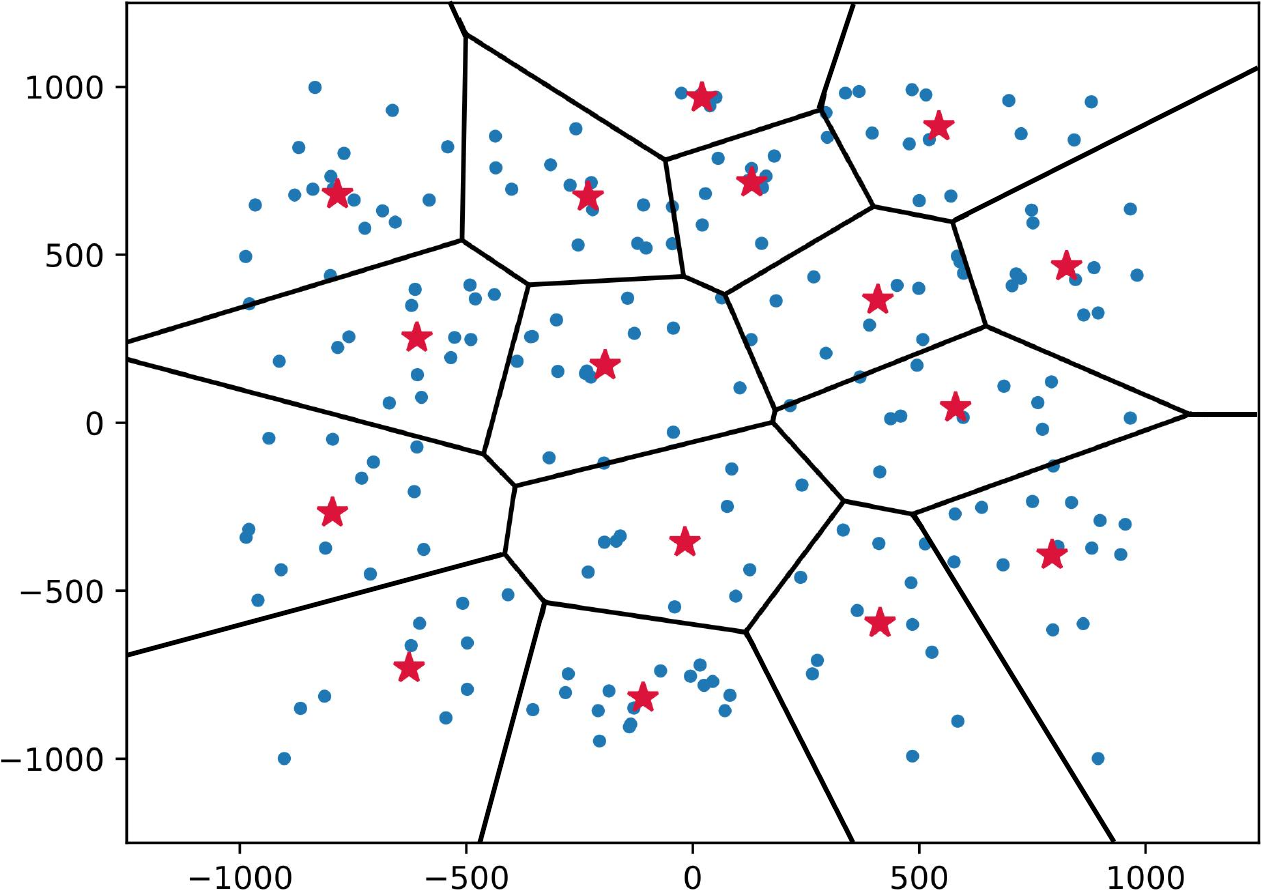}
		\label{kmeans_voronoi}}\hspace{12mm}
	\subfigure[K-means clustering with inverted multi-index. K-means clustering is performed in each dimension. The inverted multi-index holds $\sqrt{K}$ 1D centroids for each dimension. $K$ 2D clusters are obtained through Cartesian product of 1D clusters.]{
		\includegraphics[width=0.38\linewidth]{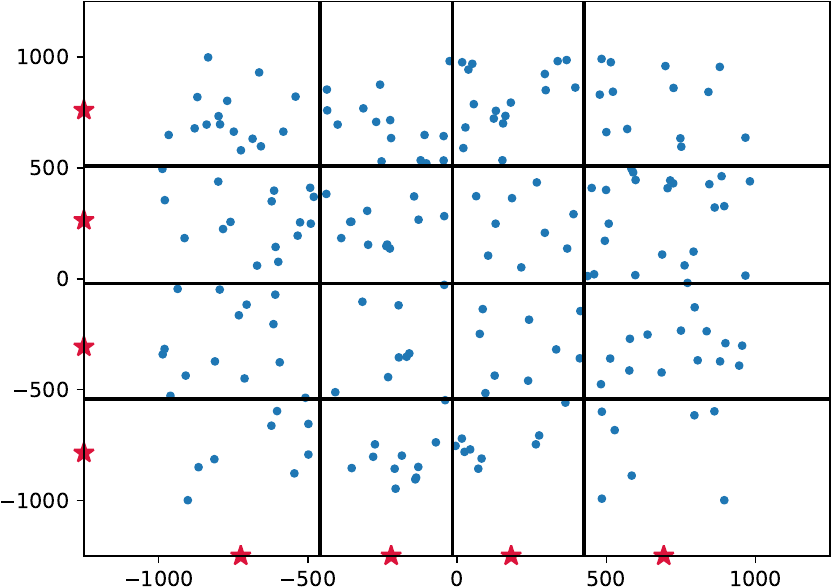}
		\label{kmeans_imi}}
	\caption{Illustration of K-means clustering ($K$=16) in 2D space using inverted index or inverted multi-index.}
	\label{kmeans_index}
\end{figure*}

\subsection{Index Construction} \label{indexing_phase}

To count collisions as quickly and accurately as possible, i.e., find the $\alpha \cdot n$ data points closest to the query in each subspace, we need to construct an index in each subspace, which should be able to quickly locate data points around the query. 

There are two requirements for the index structure: 
(1) The index structure should be lightweight enough. The index needs to be constructed independently in all $N_s$ subspaces, and $N_s$ is 6-12 in practice. 
If the index structure is complex, it requires a long indexing time and a large memory footprint to construct $N_s$ indexes. 
Therefore, the graph-based and tree-based indexes that are currently popular in ANN search are no longer applicable.
(2) The query method for collision counting needs to be simple and fast enough. 
Collision counting needs to be performed in all $N_s$ subspaces, so the query speed of the index for each subspace needs to be fast enough. 
Fortunately, we don't need to obtain few (usually $k \in \left[1, 100\right]$) but accurate query results like graph-based and tree-based indexes. 
We only need to find many ($\alpha \cdot n$, usually $\alpha \in \left[0.01, 0.1\right]$) but not necessarily accurate data points as collisions.

Clustering-based indexes can meet the above two requirements. 
The index structure is lightweight enough because only the cluster centroids and data point assignments for each cluster need to be recorded. 
The query method for collision counting is simple and fast enough because only the distance between the query and each centroid needs to be calculated. 
Each centroid can represent all data points in its cluster, and if the centroid is close enough to the query, all data points within the cluster can be considered to collide with the query.
Ordinary K-means clustering uses an inverted index to store which data points are in each cluster. 
However, to achieve fine-grained indexing, many clusters need to be constructed ($K$ is large), which will restrict indexing and query answering efficiency. 
The inverted multi-index (IMI)~\cite{babenko2014inverted} replaces the vector quantization inside the inverted index of K-means with the product quantization~\cite{jegou2010product}. 
Figure~\ref{kmeans_index} shows the difference between K-means clustering using the inverted index and the inverted multi-index. 
In Figure~\ref{kmeans_voronoi}, K-means clustering is performed in the 2D space, and the inverted index holds $K$=16 2D centroids and the IDs of all data points belonging to each cluster. 
In Figure~\ref{kmeans_imi}, K-means clustering is performed in each dimension, and the inverted multi-index holds $\sqrt{K}=4$ 1D centroids for each dimension and the IDs of all data points belonging to the cluster. Then, $K$=16 2D clusters can be obtained through the Cartesian product of 1D clusters.
Therefore, given the number of clusters $K$, K-means clustering using the inverted multi-index can reduce the time complexity from $\mathcal{O}(K \cdot n \cdot d \cdot t)$ to $\mathcal{O}(\sqrt{K} \cdot n \cdot d \cdot t)$ compared to using the inverted index, $t$ is the number of iterations.

\begin{algorithm}[tb]
	\caption{Create Index}                                                                           
	\label{create_index}
	\LinesNumbered
	\KwIn{A dataset $\mathcal D$, dataset size $n$, dimensionality of data points $d$, subspace number $N_s$, number of K-means clusters $K$, number of K-means iterations $t$}
	\KwOut{Centroid list $centroids$ and inverted multi-index list $IMIs$}
        Initialize a list $centroids$ of length $2N_s$ and a list $IMIs$ of length $N_s$; \\
	Divide the $d$-dimensional space into $N_s$ subspaces: $S_1, S_2,\ldots, S_{N_s}$; \\
        Divide all data points into $N_s$ subspaces: $\mathcal D_1,\ldots,\mathcal D_{N_s}$; \\
	\For{$i=1$ to $N_s$}{
            Initialize a map $IMI_i$ as the inverted multi-index; \\
            Further divide $S_i$ into two subspaces $S^1_i$ and $S^2_i$; 
            \\ Divide $\mathcal D_i$ into two subspaces: $\mathcal D^1_i$ and $\mathcal D^2_i$; \\
            $centroids^1_i, assignments^1_i \leftarrow$ \textbf{call} Kmeans$(\mathcal D^1_i, \sqrt{K}, t)$; \\
            $centroids^2_i, assignments^2_i \leftarrow$ \textbf{call} Kmeans$(\mathcal D^2_i, \sqrt{K}, t)$; \\
            $centroids.append(centroids^1_i, centroids^2_i)$; \\
            \For{$j=0$ to $n-1$}{
            $IMI_i[assignments^1_i[j], assignments^2_i[j]].append(j)$; \\
            }
            $IMIs.append(IMI_i)$; \\
	}
	\Return $centroids$ and $IMIs$; \\
\end{algorithm}

Figure~\ref{overview_index} shows the workflow of SuCo to construct the index, and Algorithm~\ref{create_index} gives the pseudocode.
First, we need to initialize two lists, one to store the centroids obtained by clustering in each subspace (the first part of our index) and the other to store the inverted multi-indexes constructed in each subspace (the second part of our index) (line 1). 
Then, we divide the $d$-dimensional original space into $N_s$ subspaces ($S_1,S_2,\ldots,S_{N_s}$) with $s=\frac{d}{N_s}$ dimensions, and divide all data points into each subspace (lines 2-3). 
Assume without loss of generality that $\frac{s}{2}$ is an integer.
For each subspace, based on the product quantization concept of IMI, we further divide each $s$-dimensional subspace into two $\frac{s}{2}$-dimensional subspaces (lines 4-7). 
Although each subspace can be divided into more than two parts, both the IMI paper~\cite{babenko2014inverted} and our experiments show that dividing each space into two parts achieves the best performance. 
Then, for $S_i$, we perform K-means clustering on all data points in two $\frac{s}{2}$-dimensional subspaces ($S^1_i$ and $S^2_i$) to obtain the centroids of clusters $centroids^1_i, centroids^2_i$ and the data point assignments (belongs to which cluster) $assignments^1_i, assignments^2_i$ (lines 8-9). 
We put all the obtained centroids into a list as the first part of our index (line 10). 
To improve query efficiency, in each subspace $S_i$, we use the obtained $assignments^1_i, assignments^2_i$ to construct a map $IMI_i$ as the inverted multi-index, and we put all the obtained IMIs into a list as the second part of our index (lines 11-13).
Finally, the centroids and IMI lists are returned as the SuCo indexes.

The original intention of IMI is to achieve as fine-grained data points division as possible, so $K$ is very large ($K=2^{28}$ is used in its paper~\cite{babenko2014inverted}). 
In contrast, in the subspace collision framework, we only need to find a large number of data points (3\%-5\%) in the IMI of each subspace in a coarse-grained way, so $K$ is much smaller ($K \in \left[2^{10}, 2^{12}\right]$).
That is why SuCo's indexing performance outperforms all competitors (shortest time and least memory footprint); the experimental results are shown in Section~\ref{comparison}.

\subsection{Query Answering} \label{query_phase}

To support ANN search based on the subspace collision framework, two query strategies need to design: 
(1) how to count collisions for the IMI of each subspace; 
(2) how to combine the collisions of IMIs in all subspaces to obtain the final ANN result.

The \textit{Multi-sequence} algorithm was proposed with IMI to query IMI and obtain data points that are close to the query~\cite{babenko2014inverted}. 
However, the \textit{Multi-sequence} algorithm uses a priority queue to hold the candidate clusters, 
For the priority queue, both insertion and popping operations require logarithmic time complexity, which is time-consuming. 
Therefore, we design a new algorithm called \textit{Dynamic Activation} to support query IMI without a priority queue.
\textit{Dynamic Activation} algorithm returns the same query results as \textit{Multi-sequence} algorithm.
The experimental results in Section~\ref{imiquery_compare} show that the efficiency of \textit{Dynamic Activation} algorithm is up to 40\% higher than that of the \textit{Multi-sequence} algorithm.

\begin{figure}[tb] 
	\centering
	\includegraphics[width=0.95\linewidth]{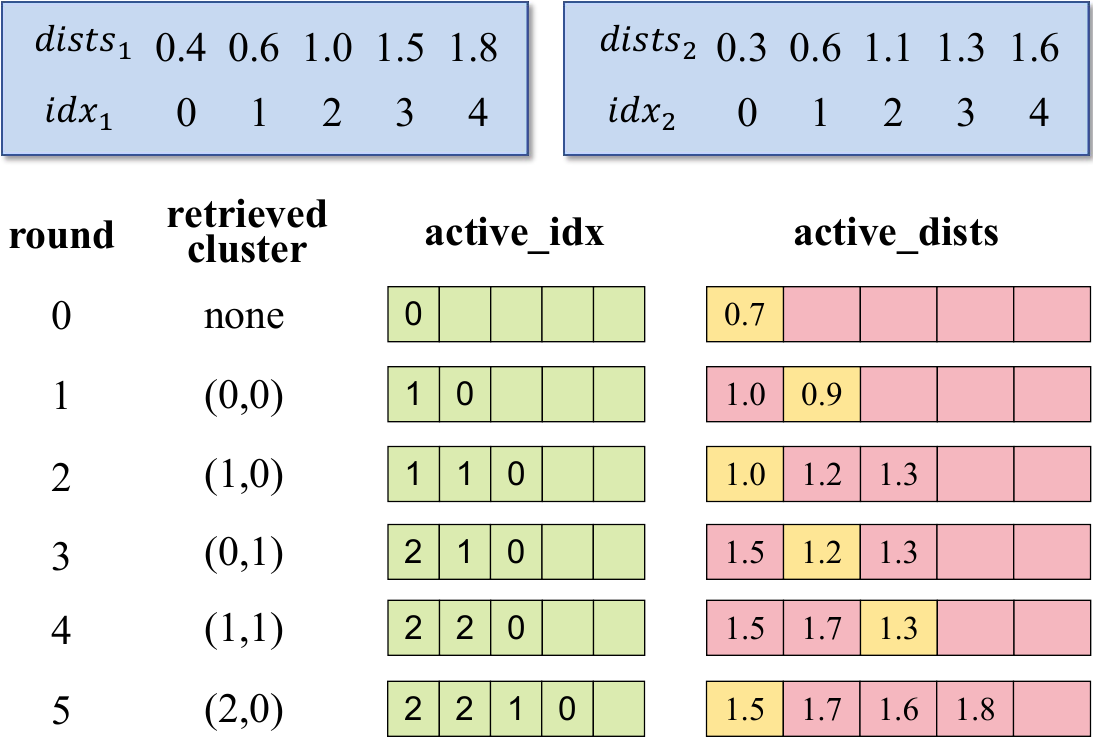}
	\caption{An illustration of the \textit{Dynamic Activation} algorithm}
	\label{dynamic_activation_example}
\end{figure}

\begin{algorithm}[t]
	\caption{Dynamic Activation}     
	\label{dynamic_activation}
	\LinesNumbered
	\KwIn{Collision ratio $\alpha$, dataset size $n$, number of K-means clusters $K$, distances and indices of the first and second subspace $dists_1,idx_1,dists_2,idx_2$, the inverted multi-index $IMI$}
	\KwOut{Clusters containing data points that collide with $q$}
        Initialize a list $retrieved\_clusters$, two arrays $active\_idx$ and $active\_dists$ of length $\sqrt{K}$; \\
        $retrieved\_num \leftarrow 0$; \\
        $active\_idx[0] \leftarrow 0$; \\
        $active\_dists[0] \leftarrow dists_1[idx_1[0]]+dists_2[idx_2[0]]$; \\
        \While{\textit{TRUE}}{
        $pos \leftarrow$ index of the minimum element in $active\_dists$; \\
        $cluster \leftarrow IMI[idx_1[pos], idx_2[active\_idx[pos]]]$; \\
        $retrieved\_clusters.append(cluster)$; \\
        $retrieved\_num$ += $sizeof(cluster)$; \\
        \If{$retrieved\_num \geq \alpha \cdot n$}{
            break; \\
        }
        \If{$active\_idx[pos]==0$ and $pos<\sqrt{K}-1$}{
            $active\_idx[pos+1] \leftarrow 0$; \\
            $active\_dists[pos+1] \leftarrow dists_1[idx_1[pos+1]]+dists_2[idx_2[0]]$; \\
        }
        \If{$active\_idx[pos]<\sqrt{K}-1$}{
            $active\_idx[pos]$++; \\
            $active\_dists[pos] \leftarrow dists_1[idx_1[pos]]+dists_2[idx_2[active\_idx[pos]]]$; \\
        }
        }
        
	\Return $retrieved\_clusters$; \\
\end{algorithm}

Algorithm~\ref{dynamic_activation} gives the pseudocode of the \textit{Dynamic Activation} algorithm.
Given the distances between a query and all centroids in two subspaces $dists_1,dists_2$ and the indices of sorted distances in ascending order $idx^{i}_1,idx^{i}_2$, we need to select one cluster from each subspace to form a cluster in $IMI$. 
The sum of their distances (one from $dists_1$, the other from $dists_2$) is used as the selection criterion. 
The smaller the distance sum of clusters in $IMI$, the earlier they are retrieved.
We first activate a cluster in the first subspace, which has the minimum distance in $dists_1$ (lines 3-4).
Activated clusters can be combined with clusters from the second subspace, and the two clusters with the minimum sum of distances will be combined as a retrieved cluster in $IMI$ (lines 6-9). 
When the number of data points in all retrieved clusters reaches the collision requirement (lines 10-11), the algorithm stops and returns the retrieved clusters (line 18).
Otherwise, we dynamically activate the clusters in the first subspace in ascending order of $dists_1$ (lines 12-14) and update the combination information of the activated clusters (lines 15-17).

Figure~\ref{dynamic_activation_example} gives a running example of Algorithm~\ref{dynamic_activation}.
Algorithm~\ref{dynamic_activation} runs in a multi-round way, and in each round, it retrieves a cluster whose sum of $dists_1$ (among the activated indices) and $dists_2$ is minimal.
In the initialization round (round 0), the index 0 of $dists_1$ is activated (with initial value 0).
In round 1, since only index 0 of $dists_1$ is activated, the cluster with the joint index of $dists_1$ and $dists_2$ (0, 0) is retrieved, and its distance to the query is 0.7.
Then, we activate index 1 of $dists_1$ (with initial value 0) and increase the value of index 0 to 1.
In round 2, two indices of $dists_1$ have been activated (0 and 1), and we retrieve the cluster with the joint index (1, 0) because it has a smaller distance (0.9) than the cluster with the joint index (0, 1).
Then, we activate index 2 of $dists_1$ (with initial value 0) and increase the value of index 1 to 1.
In round 3, three indices of $dists_1$ have been activated (0, 1, and 2), and we retrieve the cluster with the joint index (0, 1) because it has a minimum distance of 1.0.
It should be noted that no new index of $dists_1$ is activated in this round, because for the retrieved cluster, the value of the activated index is not 0 but 1.
The operations for subsequent rounds are the same as the previous rounds.

Algorithm~\ref{dynamic_activation} supports collision counting for the IMI of each subspace. 
We further design Algorithm~\ref{kann_query} to combine the collisions of IMIs in all subspaces and support $k$-ANN queries, and Figure~\ref{overview_query} shows the workflow of SuCo to answer $k$-ANN queries.
First, we initialize an array to record the SC-score of all data points and divide the query $q$ into $N_s$ subspaces (lines 1-2). 
For each subspace, we calculate the distance between the query and all centroids in the two divided parts and then obtain the indices of sorted distances in ascending order (lines 3-8).
Then, we call Algorithm~\ref{dynamic_activation} to obtain clusters containing data points that collide with $q$ in each subspace and count collisions for all data points (lines 9-12). 
Finally, we re-rank the $\beta \cdot n$ data points with the largest SC-scores and return the \textit{top}-$k$ points closest to $q$ among them (lines 13-17).

\begin{algorithm}[t]
	\caption{$k$-ANN Query}        
	\label{kann_query}
	\LinesNumbered
	\KwIn{A dataset $\mathcal D$, dataset size $n$, dimensionality of data points $d$, a query point $q$, number of results $k$, subspace number $N_s$, collision ratio $\alpha$, re-rank ratio $\beta$, number of K-means clusters $K$, the centroid list $centroids$, the IMI list $IMIs$}
	\KwOut{$k$ nearest points to $q$ in $\mathcal D$}
        Initialize an array $SC\_scores$ of length $n$ and set to 0; \\
        Divide $q$ into $N_s$ subspaces: $q^1,\ldots,q^{N_s}$; \\
	\For{$i=1$ to $N_s$}{
            Divide $q^{i}$ into two subspaces $q^{i}_1$ and $q^{i}_2$; \\
            $centroids_1 \leftarrow centroids[2(i-1)]$; \\
            $centroids_2 \leftarrow centroids[2(i-1)+1]$; \\
            Calculate the distance between each centroid in $centroids_1$($centroids_2$) and $q^{i}_1$($q^{i}_2$), obtain $dists^{i}_1$($dists^{i}_2$); \\
            Obtain the indices of sorted $dists^{i}_1$($dists^{i}_2$) in ascending order: $idx^{i}_1$($idx^{i}_2$); \\
            $clusters \leftarrow$ \textbf{call} DynamicActivation $(\alpha,n,K,dists^{i}_1,idx^{i}_1,dists^{i}_2,idx^{i}_2,IMIs[i-1])$; \\
            \For {each $cluster \in clusters$}{
                \For{each $point\_id \in cluster$}{
                    $SC\_scores[point\_id]$++; \\
                }
            }
	}
        Sort $SC\_scores$ in descending order; \\
        \For{$z=1$ to $\beta \cdot n$}{
            Select the point with the $z$-th-largest SC-score in $SC\_scores$ whose index in $\mathcal D$ is $t$; \\
            Calculate the Euclidean distance between $o_t$ and $q$; \\
        }
	\Return the \textit{top}-$k$ points closest to $q$ in the $\beta \cdot n$ candidates; \\
\end{algorithm}

\subsection{Complexity Analysis}

For index construction, SuCo has time cost $\mathcal{O}(n(\sqrt{K}dt+N_s))$ and space cost $\mathcal{O}(\sqrt{K}d+nN_s)$.
The time cost comes from three parts: 
(1) divide all data points into $N_s$ subspaces, $\mathcal{O}(nd)$; 
(2) use the K-means algorithm to cluster all data points within two parts of each subspace, $\mathcal{O}(\sqrt{K} n d t)$; 
(3) build a map as the inverted multi-index in each subspace, $\mathcal{O}(n N_s)$.
Therefore, the total time cost is $\mathcal{O}(n(\sqrt{K}dt+N_s))$.
The space cost comes from two parts: 
(1) centroids obtained by clustering in each subspace, $\mathcal{O}(\sqrt{K} d)$; 
(2) the inverted multi-indexes constructed in each subspace, $\mathcal{O}(n N_s)$.
Therefore, the total space cost is $\mathcal{O}(\sqrt{K}d+nN_s)$.

For query answering, SuCo has time cost $\mathcal{O}(n(\alpha N_s+\beta d+\log (\beta n))+\sqrt{K}(d+\alpha \sqrt{K} \log \sqrt{K}))$.
The time cost comes from seven parts: 
(1) calculate the distance between all centroids and the query point in each subspace, $\mathcal{O}(\sqrt{K} d)$;
(2) obtain the indices of sorted distances in ascending order in each subspace, $\mathcal{O}(\sqrt{K} \log \sqrt{K})$;
(3) invoke the \textit{Dynamic Activation} algorithm to obtain clusters containing data points that collide with the query in each subspace, $\mathcal{O}(\alpha K \log \sqrt{K})$;
(4) count collisions for all data points in each subspace, $\mathcal{O}(\alpha n N_s)$;
(5) partial sort and select candidates with large SC-score, $\mathcal{O}(n \log (\beta n))$;
(6) calculate the distance between all candidates and the query, $\mathcal{O}(\beta n d)$;
(7) partial sort and return the \textit{top}-$k$ points closest to the query among candidates, $\mathcal{O}(\beta n \log k)$.
Therefore, the total time cost is $\mathcal{O}(n(\alpha N_s+\beta d+\log (\beta n))+\sqrt{K}(d+\alpha \sqrt{K} \log \sqrt{K}))$.

\begin{table}
	\centering
	\caption{Summary of datasets}
	\label{datasets}
	\begin{tabular}{cccc}
		\toprule
		\textbf{Dataset} & \textbf{Cardinality} & \textbf{Dimensions} & \textbf{LID} \\
		\midrule
		Deep1M & 1,000,000 & 256 & 37.26\\
        Gist1M & 1,000,000 & 960 & 70.15\\
		Sift10M & 10,000,000 & 128 & 22.05\\
        Microsoft SPACEV10M & 10,000,000 & 100 & 41.72\\
        Yandex Deep10M & 10,000,000 & 96 & 29.10\\
		TinyImages80M & 79.302,017 & 384 & 61.75\\
		Sift100M & 100,000,000 & 128 & 23.79\\
		Yandex Deep100M & 100,000,000 & 96 & 29.61\\
		\bottomrule
	\end{tabular}
\end{table}

\section{Experimental Evaluation}

In this section, we first compare the performance of the \textit{Dynamic Activation} algorithm we designed for IMI with the original \textit{Multi-sequence} algorithm.
Then, we study the performance of SuCo by self-evaluating and conducting comparative experiments with state-of-the-art ANN methods (with and without theoretical guarantees). 
All methods are implemented in C/C++. 
SuCo is implemented in C++ and compiled using -O3 optimization, useing OpenMP for parallelization and SIMD for accelerating calculations~\cite{sucocode}.
All experiments are conducted on a machine with 2 AMD EPYC 9554 CPUs @ 3.10GHz and 756 GB RAM, running on Ubuntu 22.04.

\subsection{Experimental Setup}

\noindent \textbf{Datasets and Queries.} 
We use eight real-world datasets for ANN search, whose key statistics are shown in Table \ref{datasets}. 
LID\footnote{There is no unified method to estimate the LID, and we use an open-source toolkit~\cite{bac2021scikit} with theoretical basis~\cite{johnsson2014low} to calculate it for each dataset.} is the local intrinsic dimensionality and a larger LID implies harder dataset.
Note that the points in Sift10M and Sift100M are randomly chosen from the Sift1B dataset\footnote{http://corpus-texmex.irisa.fr/}. 
Similarly, the points in Microsoft SPACEV10M, Yandex Deep10M, and Yandex Deep100M are also randomly chosen from their 1B-scale datasets\footnote{https://big-ann-benchmarks.com/neurips21.html}. 
We randomly select 100 data points as queries and remove them from the original datasets.

\noindent \textbf{Evaluation Measures.}
We adopt six measures to evaluate the performance of all methods: indexing time, index memory footprint, query time, query per second (QPS), recall, and mean relative error (MRE)~\cite{patella2008many,patella2009approximate,aumuller2020ann}, where the indexing time, query time, and QPS evaluate the efficiency of methods, the index memory footprint evaluates the storage resource consumption of methods, and the recall and MRE evaluate the quality of returned results. 
For a query $q$, if the returned result set is $R=\{o_1,\ldots,o_k\}$ and the exact $k$-NN set is $R^*=\{o_1^*,\ldots,o_k^*\}$, recall is defined as $\frac{\lvert R \cap R^* \rvert}{k}$, and mean relative error (MRE) is defined as $\frac{1}{k} \sum_{i=1}^{k} \frac{\left\|q,o_i\right\|-\left\|q,o_i^*\right\|}{\left\|q,o_i^*\right\|}$.

\noindent \textbf{Benchmark Methods.}
We compare SuCo with seven state-of-the-art in-memory ANN methods. 
DET-LSH~\cite{detlsh}, DB-LSH~\cite{dblsh}, PM-LSH~\cite{pmlsh}, and LCCS-LSH~\cite{lccslsh} are the state-of-the-art LSH-based methods that provide theoretical guarantees.
They are single-threaded because the authors did not design parallel methods.
OPQ~\cite{ge2013optimized}, Annoy~\cite{annoy}, HNSW~\cite{malkov2018efficient}, and SPTAG~\cite{chen2018sptag} are the state-of-the-art VQ-based, tree-based, graph-based, and tree-graph hybrid methods, respectively, which do not provide theoretical guarantees.
They use OpenMP and/or Pthreads for parallelization, and SIMD for accelerating calculations.


\begin{figure}[tb] 
	\centering
	\includegraphics[width=0.95\linewidth]{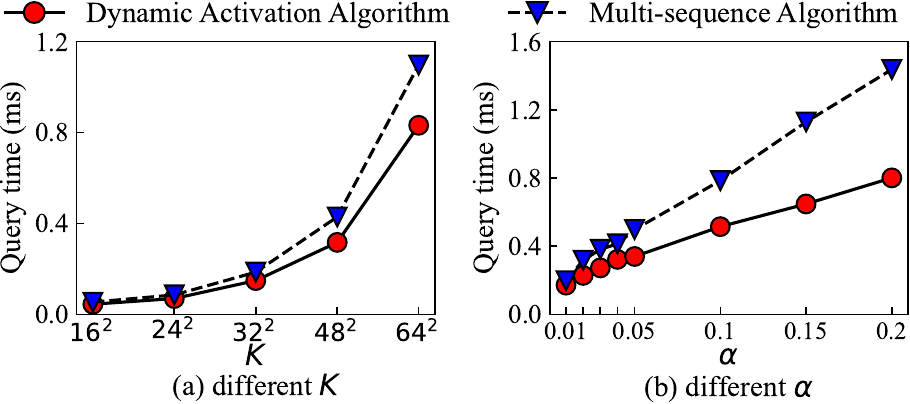}
	\caption{Comparison of query efficiency between Dynamic Activation and Multi-sequence algorithm on Sift10M.}
	\label{smi_query}
\end{figure}





\noindent \textbf{Parameter Settings.} 
Both the indexing and query answering phases of ANN methods require parameter settings.
To select indexing parameters, we refer to the optimal parameters given in previous papers~\cite{hydra2,li2019approximate,detlsh}  and conduct our own experiments to verify their rationality so that all methods can achieve the best indexing performance.
For DET-LSH, $L=4$, $K=16$. 
For DB-LSH, $L=5$, $K=12$. 
For PM-LSH, $s=5$, $m=15$.
For LCCS-LSH, $m=64$. 
For OPQ, $M=2$, $K=2^{20}$ (number of cells after Cartesian product).
For Annoy, $f=30$.
For HNSW, $efConstruction=200$, $M=25$.
For SuCo, $K=50^2$, $N_s=8$.
For query parameters, we dynamically adjust their settings to observe the trade-off between query efficiency and accuracy (as shown in Figure~\ref{recallqps_guarantee} and Figure~\ref{recallqps_all}).
For DET-LSH, DB-LSH, and PM-LSH, $\beta \in \left[0.005,0.2\right]$, $c=1.5$. 
For LCCS-LSH, $check\_k \in \left[2^8,2^{18}\right]$.
For OPQ, $\beta \in \left[10^{-4},10^{-1}\right]$.
For Annoy, $search\_k \in \left[10^3,10^5\right]$.
For HNSW, $efSearch \in \left[300,3000\right]$.
For SPTAG, we use the balanced K-means tree (BKT).
For SuCo, $\alpha \in \left[0.01, 0.1\right]$, $\beta \in \left[0.001, 0.05\right]$.
$k$ in $k$-ANN is set to 50.

\subsection{Dynamic Activation vs. Multi-sequence} \label{imiquery_compare}
As introduced in Section~\ref{query_phase}, we design a new query strategy (Algorithm~\ref{dynamic_activation}) to support efficient querying of IMI. 
Figure~\ref{smi_query} compares the query efficiency between \textit{Dynamic Activation} algorithm (we designed) and \textit{Multi-sequence} algorithm (proposed with IMI~\cite{babenko2014inverted}). 
While achieving the same query accuracy (they return the same query results), the efficiency of \textit{Dynamic Activation} algorithm is up to 40\% higher than that of the \textit{Multi-sequence} algorithm.
The advantage of \textit{Dynamic Activation} algorithm increases as the query workload increases (larger $K$ and $\alpha$).
The reason is that the \textit{Multi-sequence} algorithm relies on a priority queue to hold the candidate clusters, and frequent insertion and popping operations are time-consuming.
\textit{Dynamic Activation} algorithm updates and maintains the activation list through activation strategies, eliminating the reliance on the priority queue and improving query efficiency.

\begin{table}
	\centering
	\caption{Comparison on SuCo and SC-Linear}
	\label{scevaluation}
\begin{tabular}{@{}ccccc@{}}
\toprule
\textbf{Method}                               & \textbf{Dataset} & \textbf{Query Time (ms)} & \textbf{Speedup} & \textbf{Recall} \\ \midrule
\multirow{2}{*}{SC-Linear}                    & Sift10M          & 3104.42                  & \textbackslash{} & 0.968           \\
                                              & Sift100M         & 71644.5                  & \textbackslash{} & 0.9942           \\ 
\multicolumn{1}{c}{\multirow{2}{*}{SuCo}} & Sift10M          & 5.139                    & 604.1            & 0.9346          \\
\multicolumn{1}{c}{}                        & Sift100M         & 68.835                   & 1040.8            & 0.9822          \\ \bottomrule
\end{tabular}
\end{table}

\begin{figure*}[tb] 
	\centering
	\includegraphics[width=0.95\linewidth]{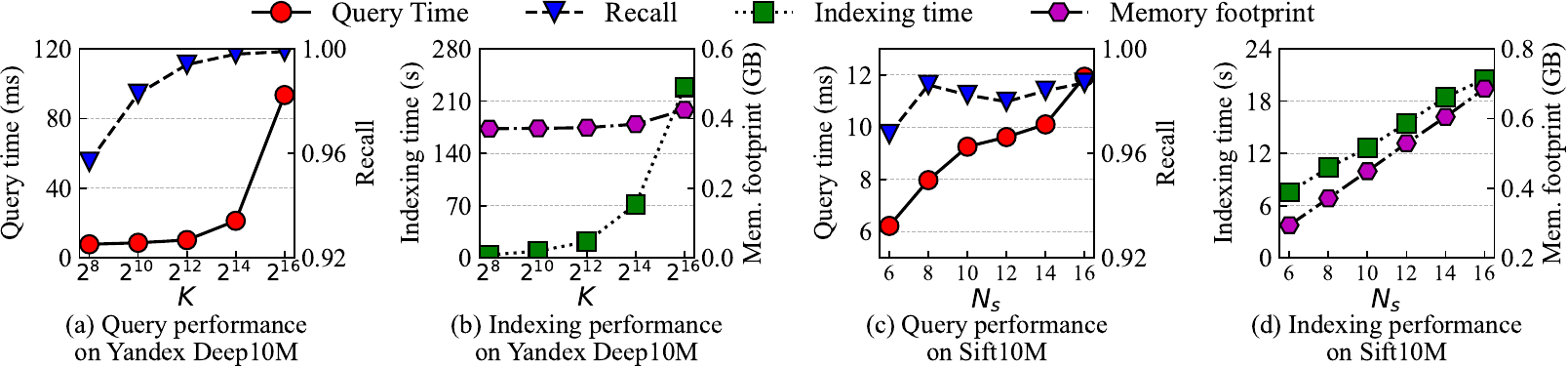}
	\caption{Performance of SuCo when varying the number of K-means clusters $K$ and the number of subspaces $N_s$.}
	\label{diff_K_Ns}
\end{figure*}

\begin{figure*}[tb] 
	\centering
	\includegraphics[width=0.95\linewidth]{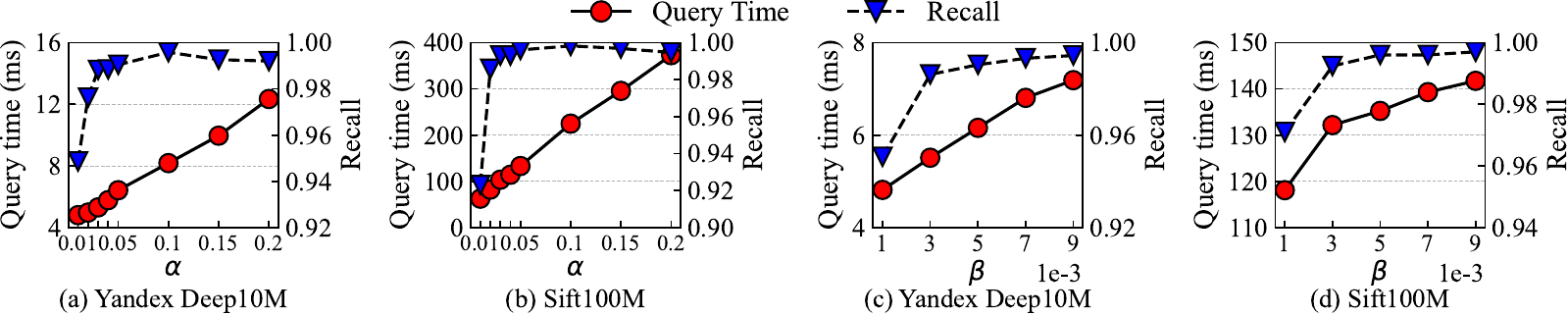}
	\caption{Query performance of SuCo when varying the collision ratio $\alpha$ and the re-rank ratio $\beta$.}
	\label{diff_alpha_beta}
\end{figure*}

\subsection{Self-evaluation of SuCo} \label{selfevaluation}

\subsubsection{SuCo vs. SC-Linear}
Preliminary experiments in Section~\ref{preliminary_experiments} have demonstrated that the subspace collision framework is effective, and Table~\ref{SC-Linear_recall} shows that SC-Linear (has no index) can return high-quality query results. 
In this section, we further compare the query performance between SuCo and SC-Linear under the same parameter settings: $\alpha=0.03$, $\beta=0.003$.
As shown in Table~\ref{scevaluation}, compared with SC-Linear, the query efficiency of SuCo is significantly improved (up to 1000 times), with an acceptable sacrifice of query accuracy.
The results demonstrate that the index structure and query strategy we designed for SuCo can be well applied to the subspace collision framework for ANN search.


\begin{table}
	\centering
	\caption{SuCo under different distance measures}
	\label{distance_measures}
\begin{tabular}{@{}cccccc@{}}
\toprule
\textbf{}                               & \textbf{} & \makecell[c]{Gist\\1M} & \makecell[c]{Sift\\10M} & \makecell[c]{SPACEV\\10M} & \makecell[c]{Yandex Deep\\100M} \\ \midrule
\multicolumn{1}{c}{\multirow{2}{*}{SuCo-L1}}                    & Recall          & 0.9328                  & 0.9868 & 0.966 & 0.9906           \\
                                              & MRE         & 0.00099                  & 0.00039 & 0.00103 & 0.0002           \\ 
\multicolumn{1}{c}{\multirow{2}{*}{SuCo-L2}} & Recall          & 0.9508                    & 0.9862    & 0.9802           & 0.9988          \\
\multicolumn{1}{c}{}                        & MRE         & 0.00068                   & 0.00042  & 0.00035            & 0.00002          \\ \bottomrule
\end{tabular}
\end{table}

\subsubsection{Parameter study on $K$ and $N_s$}
Parameters $K$ and $N_s$ determine the efficiency and space consumption of index construction and affect the efficiency and accuracy of queries.
We evaluate the impact of $K$ and $N_s$ on the performance of SuCo by setting $K$ in the range $\left[2^8, 2^{16}\right]$ and $N_s$ in the range $\left[6, 16\right]$.
As shown in Figure~\ref{diff_K_Ns} (a) and (b), the indexing time and query time increase gradually with $K$ increasing from $2^8$ to $2^{12}$ but significantly with $K$ increasing from $2^{12}$ to $2^{16}$.
In contrast, the recall increase significantly with $K$ increasing from $2^8$ to $2^{12}$ but gradually with $K$ increasing from $2^{12}$ to $2^{16}$.
The index memory footprint remains rather stable in the whole range.
Therefore, we choose $K=50^2 \in (2^{10}, 2^{12})$ as the default value.
As shown in Figure~\ref{diff_K_Ns} (c) and (d), the indexing time, index memory footprint, and query time continue to increase as $N_s$ increases.
The recall remains rather stable when $N_s$ is larger than 8.
Considering these two factors, we choose $N_s=8$ as the default value.
Based on our experiments with datasets of different dimensionality and size, we suggest choosing $N_s \in \left[6,12\right]$ for new datasets, which is a good balance between indexing cost and query accuracy.

\begin{figure}[tb] 
	\centering
	\includegraphics[width=\linewidth]{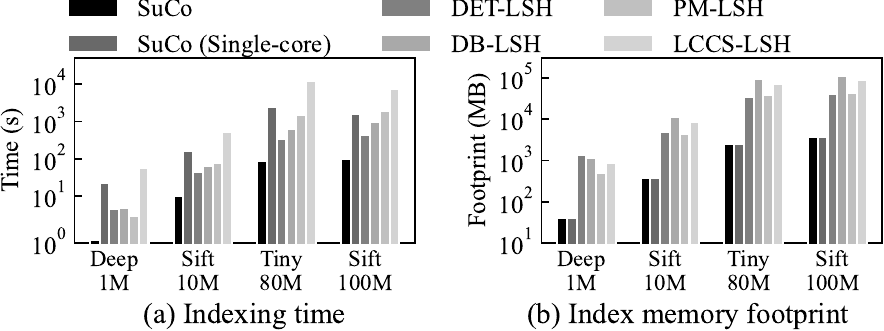}
	\caption{Indexing performance comparison between SuCo and competitors that provide theoretical guarantees.}
	\label{indexing_lsh}
\end{figure}

\subsubsection{Parameter study on $\alpha$ and $\beta$}
\label{sec:parameter_study}

Parameters $\alpha$ and $\beta$ have a great impact on query efficiency and accuracy.
We evaluate the impact of $\alpha$ and $\beta$ on the query performance of SuCo by setting $\alpha$ in the range $\left[0.01, 0.2\right]$ and $\beta$ in the range $\left[0.001, 0.009\right]$.
As shown in Figure~\ref{diff_alpha_beta} (a) and (b), the query time continues to increase as $\alpha$ increases, but the recall is rather stable when $\alpha$ is larger than 0.05. 
As shown in Figure~\ref{diff_alpha_beta} (c) and (d), the query time continues to increase as $\beta$ increases, but the recall increases slowly when $\beta$ is larger than 0.005.
Therefore, $\alpha=0.05, \beta=0.005$ is a good choice for SuCo.
In addition, we found that SuCo can achieve a high recall with smaller $\alpha$ and $\beta$ on large-scale datasets. 
On the Sift100M dataset, the recall of SuCo is 0.9822 when $\alpha=0.03$, $\beta=0.003$.
This demonstrates that SuCo has good scalability on large-scale datasets.
Based on our experiments with datasets of different dimensionality and size, we suggest choosing $\alpha \in \left[0.03, 0.1\right]$ and $\beta \in \left[0.003, 0.005\right]$ for new datasets, which is a good balance between query efficiency and accuracy.
In addition, large-scale datasets tend to use smaller $\alpha$ and $\beta$, and hard datasets tend to use larger $\beta$.

\begin{figure}[tb] 
	\centering
	\includegraphics[width=\linewidth]{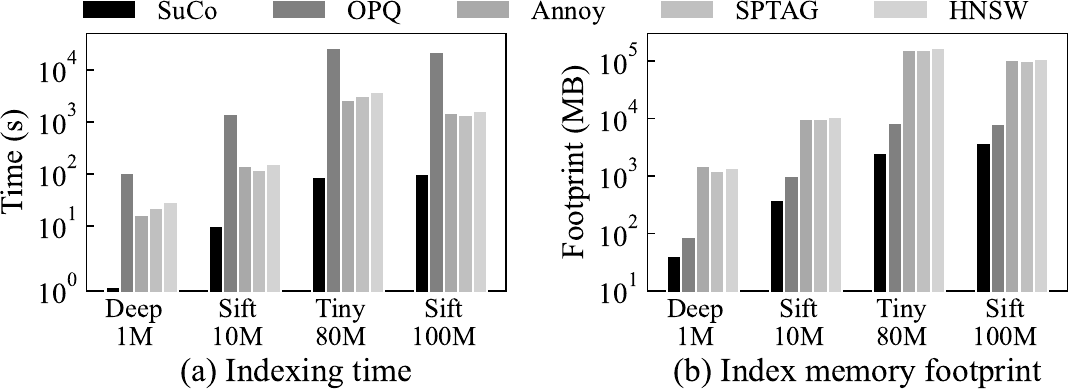}
	\caption{Indexing performance comparison between SuCo and competitors that do not provide theoretical guarantees.}
	\label{indexing_all}
\end{figure}

\begin{figure*}[tb] 
	\centering
	\includegraphics[width=0.86\linewidth]{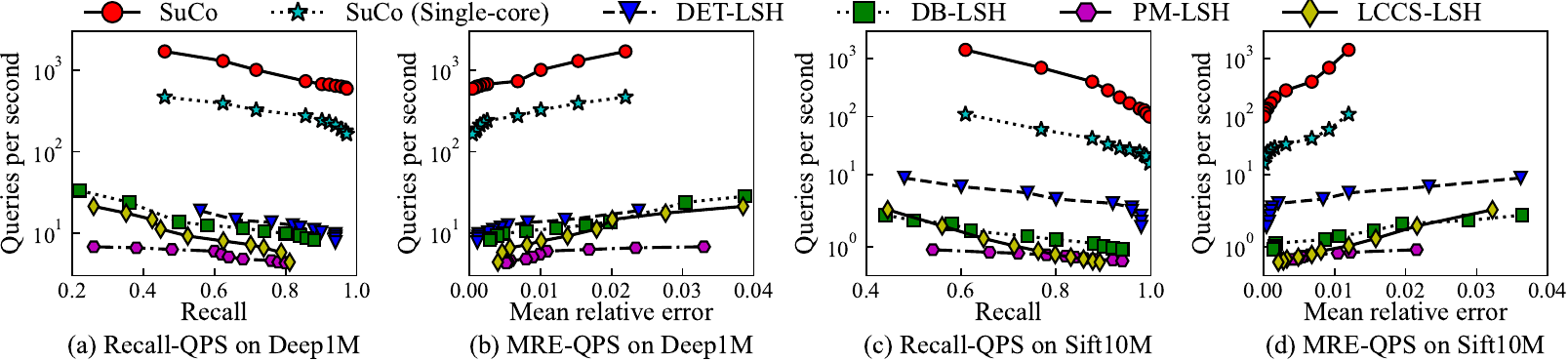}
	\caption{Query performance comparison between SuCo and competitors that provide theoretical guarantees.}
	\label{recallqps_guarantee}
\end{figure*}

\begin{figure*}[tb] 
	\centering
	\includegraphics[width=0.86\linewidth]{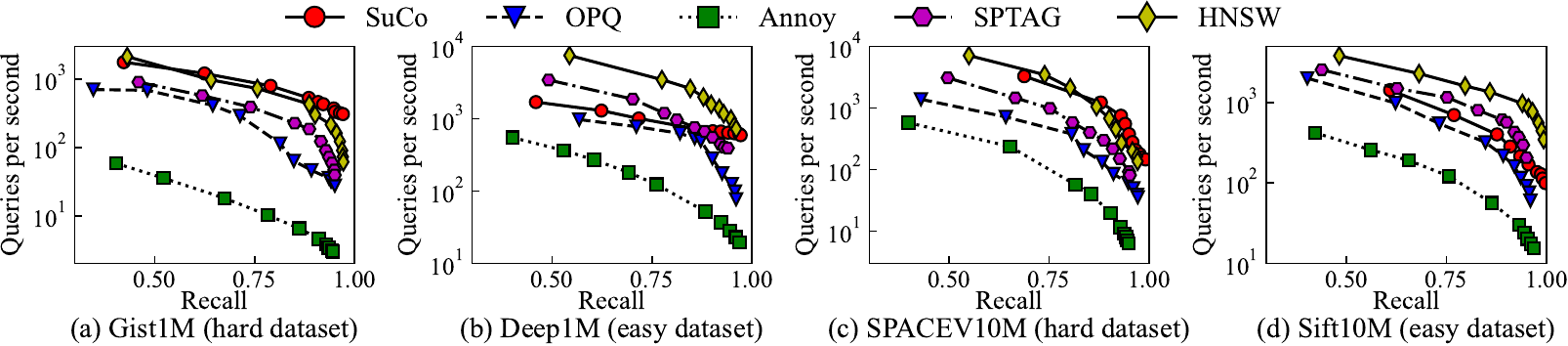}
	\caption{Query performance comparison between SuCo and competitors that do not provide theoretical guarantees.}
	\label{recallqps_all}
\end{figure*}

\begin{figure*}[tb] 
	\centering
	\includegraphics[width=0.86\linewidth]{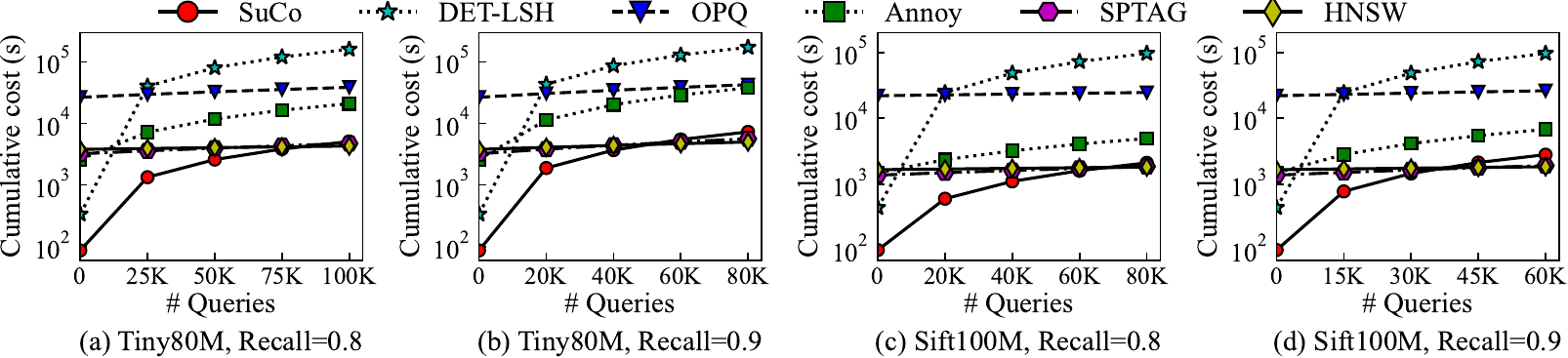}
	\caption{Cumulative query cost (start with the indexing time); comparison to methods with/without guarantees.}
	\label{query_indexing}
\end{figure*}

\begin{figure*}[tb] 
	\centering
	\includegraphics[width=0.86\linewidth]{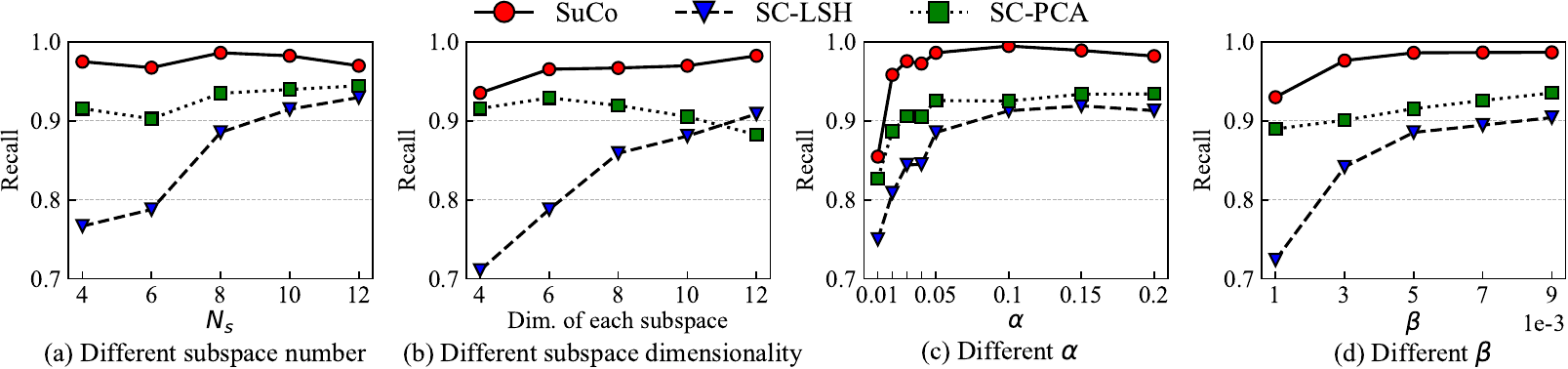}
	\caption{Query performance of SC-based methods with different data preprocessing techniques on Sift10M.}
	\label{diff_preprocessing_techniques}
\end{figure*}

\subsubsection{SuCo under different distance measures}
In this paper, we focus on the L2 distance (Euclidean distance), which is the most widely used distance measure.
It is interesting to examine if SuCo can be used with other distance measures.
Table~\ref{distance_measures} reports the performance of SuCo using the L1 (Manhattan distance) and L2 distance measures on four datasets under the same setting.
We observe that SuCo achieves high query accuracy with L1, too.



\begin{table*}
	\centering
	\caption{The time and space complexity of the state-of-the-art methods (for notation, see Section~\ref{sec:complexitycomparison})}
	\label{complexity}
\begin{tabular}{@{}ccccccc@{}}
\toprule
\textbf{}                               & SuCo & DET-LSH & OPQ & Annoy & SPTAG & HNSW\\ \midrule
Indexing time                    & $\mathcal{O}(n(\sqrt{K}dt+N_s))$          & $\mathcal{O}(LK n(d+\log N_r))$                  & $\mathcal{O}(n(\sqrt{K}+k^*)dt)$ & $\mathcal{O}(Lnd\log n)$ & $\mathcal{O}(nd\log n)$ & $\mathcal{O}(nd\log n)$          \\
Index space                                              & $\mathcal{O}(\textcolor{gray}{\sqrt{K}d}+nN_s)$         & $\mathcal{O}(LK n)$                  & \makecell[c]{$\mathcal{O}(n(m\log k^*+\log \sqrt{K})$\\$+\textcolor{gray}{(\sqrt{K}+k^*)d})$} & $\mathcal{O}(Ldn)$ & $\mathcal{O}(n(d+E))$  &  $\mathcal{O}(nE)$        \\ 
Query time & \makecell[c]{$\mathcal{O}(n(\alpha N_s+\beta d+\log (\beta n))$\\$+\textcolor{gray}{\sqrt{K}(d+\alpha \sqrt{K} \log \sqrt{K}}))$}          & $\mathcal{O}(n(\beta d+LK\log N_r))$                    & \makecell[c]{$\mathcal{O}(n(\gamma m+\beta d)$\\$+\textcolor{gray}{(\sqrt{K}+k^*)d})$}    & \makecell[c]{$\mathcal{O}(\textcolor{gray}{Ld\log n}$\\$+\beta nd)$}           & $\mathcal{O}(d\log n)$ & $\mathcal{O}(d\log n)$\\ \bottomrule
\end{tabular}
\end{table*}

\subsection{SuCo vs. Competitors with Guarantees} \label{comparison_lsh}

\subsubsection{Indexing performance}
Figure~\ref{indexing_lsh} shows the indexing time and index memory footprint of all methods that provide theoretical guarantees.
We found that SuCo and single-threaded SuCo have the least index memory footprint. Since SuCo’s clustering-based index structure is lightweight and only the cluster centroids and data point assignments need to be recorded, the required memory space is small.
In addition, SuCo has the best indexing efficiency, while single-threaded SuCo is (naturally) slower, with the extra cost coming from the (now serial) subspace clustering step.
In practice, parallel indexing methods are used in the modern multi-core systems to speed up index construction, so SuCo has an advantage in indexing efficiency compared to other methods with theoretical guarantees.
LSH-based methods need to build one (PM-LSH) or multiple (DET-LSH, DB-LSH) index trees or circular shift array (LCCS-LSH), which are heavy-weight index structures, requiring more time to construct the index, and occupying more memory.


\subsubsection{Query performance}
For all methods, there is a trade-off between query efficiency and accuracy.
Figure~\ref{recallqps_guarantee} gives the recall-QPS and MRE-QPS curves of all methods that have theoretical guarantees. 
We found that the query performance of SuCo and single-threaded SuCo is 1-2 orders of magnitude better than that of LSH-based methods. 
Benefiting from the design of the subspace collision framework, index structure, and query strategy, SuCo can efficiently count collisions to obtain candidate points and then ensure the high quality of query results through a re-ranking mechanism.
SuCo only needs to retrieve 0.3\%-0.5\% of candidate points to achieve a high recall, indicating that the obtained candidate points are of high quality. 
However, DET-LSH, DB-LSH, and PM-LSH need to obtain more candidate points (up to 10\%) and spend more time to achieve the same recall because the LSH projection loses the distance information between data points.

\subsection{SuCo vs. Competitors without Guarantees} \label{comparison}

\subsubsection{Indexing performance}
Figure~\ref{indexing_all} shows the indexing time and index memory footprint of SuCo and competitors that do not provide theoretical guarantees.
We found that SuCo has the best indexing efficiency and the least index memory footprint.
Since SuCo only needs to find many (about 3\%-5\%) but not necessarily accurate data points as collisions, the index structure does not need to divide the data points very finely.
To construct the IMI of each subspace, we only need to perform K-means clustering of $\sqrt{K}=50$ twice, which makes SuCo highly efficient for indexing.
However, the tree structure partitions data points in a fine-grained manner, while the graph structure connects each data point to its sufficiently close neighbor points, which makes the index heavyweight and requires longer indexing time and larger memory space (Annoy, SPTAG, and HNSW).
Although OPQ also uses IMI as the index structure, its query strategy requires IMI to be able to divide data points very finely (different from our subspace collision framework), so the $K$ of K-means clustering needs to be very large ($\sqrt{K}=2^{10}-2^{14}$), which seriously restricts the indexing efficiency of OPQ.

\subsubsection{Query performance}
Figure~\ref{recallqps_all} gives the recall-QPS curves of SuCo and competitors that do not provide theoretical guarantees.
We found that SuCo and HNSW outperform other methods, and on the hard datasets (Gist1M, SPACEV10M), SuCo performs better, while on the easy datasets (Deep1M, Sift10M), HNSW performs better.
Graph-based methods (such as HNSW) sacrifice indexing time and space consumption in exchange for higher query efficiency because they need to identify the near neighbors for each data point in the dataset (and connect to them) during the indexing phase, while in the query phase, they only need to search on a gradually converging path.
However, on the hard datasets, HNSW's greedy query strategy is prone to falling into local optimal subgraphs and is difficult to escape. 
In addition, HNSW requires few and accurate query results directly from the index, so its query strategy has low fault tolerance.
On the contrary, SuCo only needs to find many (about 3\%-5\%) but not necessarily accurate data points from its index as collisions, so its query strategy is highly fault-tolerant and more suitable for hard datasets.
It is important to emphasize that SuCo has to pay the cost of providing guarantees for its answers; other competitors, that do not provide any guarantees, do not pay this cost.
Even with this cost, SuCo achieves top performance (best for hard datasets) when compared to competitors that do not provide theoretical guarantees.

\subsection{Overall Evaluation}
We now evaluate the cumulative indexing and query time cost of methods with theoretical guarantees (SuCo and DET-LSH) and without (OPQ, Annoy, SPTAG, and HNSW).
Figure~\ref{query_indexing} shows the cumulative query costs of all methods, where the cost starts with the indexing time.
Compared with methods that provide theoretical guarantees (e.g., DET-LSH), SuCo always has a considerable performance advantage, ranging between $1$ and more than $2$ orders of magnitude.
Even when compared with methods that do \emph{not} provide theoretical guarantees (OPQ, Annoy, SPTAG, and HNSW), SuCo also has an advantage: it creates the index and answers 40K-80K queries before the best competitor (i.e., HNSW) answers its first query.
Note that SuCo performs better on hard datasets: compared with the Sift100M (easy) dataset, SuCo answers more queries on the Tiny80M (hard) dataset before HNSW answers its first query.


\subsection{Complexity Comparison and Analysis}
\label{sec:complexitycomparison}

Since the index structures and query strategies of the state-of-the-art methods are different (refer to Table~\ref{complexity}), it is hard to make a direct comparison among them. 
In our analysis, we will focus on the main parts that affect the complexity of each method (ignoring secondary terms marked in gray in Table~\ref{complexity}, which are not at the same complexity level as the main items, e.g., $\mathcal{O}(\log n)$ and $\mathcal{O}(n)$) and conduct a comparative analysis based on these parts.
For all methods, $n,d,\beta,t$ have the same meanings as defined in this paper.

In terms of indexing time, Annoy ($L$ is the number of trees), HNSW, and SPTAG need more than $\mathcal{O}(n\log n)$ complexity, which is time-consuming.
Although OPQ has the same complexity level $\mathcal{O}(n)$ as SuCo and DET-LSH ($L$ and $K$ are the number and dimensionality of projected spaces, and $N_r$ is the number of regions in each projected space), it requires two layers of clustering for quantization (the number of centroids is $K$ and $k^*$, respectively), so it takes a long time to build the index.
Therefore, SuCo and DET-LSH are suitable for scenarios sensitive to indexing time.
In terms of index space, SuCo and OPQ are cluster-based indexes, and their complexity coefficients $N_s$, $m$ (the number of subvectors), and $k^*$ are relatively small. 
The tree-based indexes (DET-LSH, Annoy), graph-based index (HNSW), and hybrid index (SPTAG) have larger complexity coefficients, i.e., the number of trees ($L$) and the number of neighbors for each vertex in the graph ($E$).
Therefore, SuCo and OPQ are suitable for scenarios with limited index space.
In terms of query time, HNSW and SPTAG have $\mathcal{O}(\log n)$ complexity, while other methods rely on re-ranking, with a complexity of $\mathcal{O}(\beta nd)$.
Since the subspace collision framework can effectively select candidate points, SuCo has a smaller re-rank ratio $\beta$ than DET-LSH, OPQ ($\gamma$ is the candidate pool ratio), and Annoy.
Therefore, HNSW, SuCo, and SPTAG are suitable for scenarios sensitive to query time, as validated by the experimental results in Section~\ref{comparison_lsh} and Section~\ref{comparison}.

\subsection{Impact of Data Preprocessing Techniques}

Our proposed subspace collision framework consists of two ideas: (1) use a simple division strategy to preprocess the raw data into subspaces and (2) use $\alpha$ and $\beta$ to generate a pool of NN candidates and then pick the best NNs from that pool. 
We explored the effects of combining different data preprocessing techniques (LSH for projection, Principal Component Analysis (PCA) for
dimensionality reduction, and our simple division) with the subspace collision framework, as shown in Figure~\ref{diff_preprocessing_techniques}.
Figure~\ref{diff_preprocessing_techniques} (a) and Figure~\ref{diff_preprocessing_techniques} (b) show that SuCo outperforms SC-LSH and SC-PCA in query accuracy under different subspace settings. 
In addition, SuCo's data preprocessing speed is 4x and 12x faster than SC-LSH and SC-PCA, respectively.
Figure~\ref{diff_preprocessing_techniques} (c) and Figure~\ref{diff_preprocessing_techniques} (d) explore the query performance of SC-based methods under different alpha-beta settings: SuCo exhibits the best performance.
Therefore, our proposed simple division strategy is more beneficial to the subspace collision framework than previously proposed data preprocessing techniques.

\section{Conclusions}

In this paper, we first designed SC-score, a metric that follows the \enquote{Pareto principle} and can act as a proxy for the Euclidean distance between data points.
Next, we proposed a novel ANN search framework called subspace collision, which can achieve high recall and provide theoretical guarantees on the quality of its results.
Then, we proposed SuCo, which achieves efficient and accurate ANN search by designing a clustering-based lightweight index and query strategies for the subspace collision framework.
Finally, we conducted extensive experiments, and the results demonstrate the superiority of SuCo in indexing and query answering performance.

In future work, we will combine deep learning to explore more efficient index structures under the subspace collision framework.


\begin{acks}
X.~Lee and B.~Peng partially supported by National Natural Science Foundation of China (62202450) and FUXI Institution-CASICT Interenet Infrastructure Laboratory (E051570).
Z.~Liao partially supported by National Natural Science Foundation of China (NSFC-62206101, NSFC-12141107), and Guangdong Provincial Key Laboratory of Mathematical Foundations for Artificial Intelligence (2023B1212010001).
T. Palpanas partially supported by EU Horizon projects AI4Europe (101070000), TwinODIS (101160009), ARMADA (101168951), DataGEMS (101188416) and RECITALS (101168490).

\end{acks}

\bibliographystyle{ACM-Reference-Format}
\balance
\bibliography{ref}


\end{document}